\documentclass[aps,twocolumn,
superscriptaddress,
footinbib,
prb,floatfix]{revtex4-2}

\usepackage{graphicx}
\usepackage{enumitem}
\usepackage{mwe}
\usepackage{mathtools}
\usepackage{bbold}
\usepackage{amsmath,amssymb,bm}
\usepackage{mathrsfs}
\usepackage{MnSymbol}
\usepackage{graphicx}
\usepackage{epstopdf}
\usepackage[normalem]{ulem} 
\usepackage[caption=false]{subfig}
\usepackage{comment}
\usepackage{subfig}
\usepackage[usenames,dvipsnames]{color}
\usepackage{hyperref}
\usepackage[nameinlink]{cleveref}
\usepackage{natbib}
\usepackage{xcolor, colortbl}
\usepackage{physics}
\usepackage{multirow}
\usepackage{stackengine}
\usepackage{dsfont}
\usepackage{relsize}
\usepackage{framed}
\usepackage[utf8]{inputenc} 
\usepackage[acronym]{glossaries}
\usepackage{ragged2e}
\usepackage{algorithm}
\usepackage{algpseudocode} 
\usepackage{amsthm} 

\newtheorem{prop}{Proposition}

\begin{document}


\title{Readout error mitigated quantum state tomography tested on superconducting qubits}

\date{\today}

\author{Adrian Skasberg Aasen}
\thanks{These authors contributed equally to this work.}
\affiliation{Kirchhoff-Institut f\"{u}r Physik, Universit\"{a}t Heidelberg, Im Neuenheimer Feld 227, 69120 Heidelberg, Germany}
\affiliation{Institut für Festk\"{o}rpertheorie und -optik,  Friedrich-Schiller-Universit\"{a}t Jena, Max-Wien-Platz 1, 07743 Jena, Germany}

\author{Andras Di Giovanni}
\thanks{These authors contributed equally to this work.}
\affiliation{Physikalisches Institut, Karlsruher Institut für Technologie, Kaiserstraße 12, 76131 Karlsruhe, Germany}

\author{Hannes Rotzinger}
\affiliation{Physikalisches Institut, Karlsruher Institut für Technologie, Kaiserstraße 12, 76131 Karlsruhe, Germany}
\affiliation{Institut für QuantenMaterialien und Technologien, Karlsruher Institut für Technologie, Kaiserstraße 12, 76131 Karlsruhe, Germany}

\author{Alexey V. Ustinov}
\affiliation{Physikalisches Institut, Karlsruher Institut für Technologie, Kaiserstraße 12, 76131 Karlsruhe, Germany}
\affiliation{Institut für QuantenMaterialien und Technologien, Karlsruher Institut für Technologie, Kaiserstraße 12, 76131 Karlsruhe, Germany}

\author{Martin G\"{a}rttner}
\email{martin.gaerttner@uni-jena.de}
\affiliation{Institut für Festk\"{o}rpertheorie und -optik,  Friedrich-Schiller-Universit\"{a}t Jena, Max-Wien-Platz 1, 07743 Jena, Germany}

\begin{abstract}
    Quantum technologies rely heavily on accurate control and reliable readout of quantum systems. Current experiments are limited by numerous sources of noise that can only be partially captured by simple analytical models and additional characterization of the noise sources is required. We test the ability of readout error mitigation to correct realistic noise found in systems composed of quantum two-level objects (qubits). To probe the limit of such methods, we designed a beyond-classical readout error mitigation protocol based on quantum state tomography (QST), which estimates the density matrix of a quantum system, and quantum detector tomography (QDT), which characterizes the measurement procedure. By treating readout error mitigation in the context of state tomography the method becomes largely readout mode-, architecture-, noise source-, and quantum state-independent. We implement this method on a superconducting qubit and evaluate the increase in reconstruction fidelity for QST. We characterize the performance of the method by varying important noise sources, such as suboptimal readout signal amplification, insufficient resonator photon population, off-resonant qubit drive, and effectively shortened $T_1$ and $T_2$ decay times. As a result, we identified noise sources for which readout error mitigation worked well, and observed decreases in readout infidelity by a factor of up to 30.   
\end{abstract}

\maketitle

\section{Introduction}
\label{sec:intro}

Building quantum machines capable of harnessing superposition and entanglement promises advances in many fields ranging from cryptography \cite{Bennett2014, Shor1997}, material simulation \cite{Ma2020} and drug discovery \cite{Pyrkov2023, Zinner2022} to finance \cite{Ors2019, Egger2021, Dri2022} and route optimizations \cite{Harwood2021}. While these applications assume a large-scale fault-tolerant quantum computer, currently we are still in the era of noisy intermediate scale quantum (NISQ) devices \cite{Preskill2018}, where noise \footnote{In this work we refer to noise as the physical phenomenon which causes variable operation of the experiment, and error is the deviation of the measured quantities with respect to their theoretical expectation values.} and imperfection of qubits fundamentally limit the applicability of quantum algorithms. To harness the power of quantum computation and quantum simulation today, we need methods that are suitable for working with noisy hardware.

There are predominantly two approaches to combating errors in qubit systems. On the one hand, one can study the exact sources of noise and errors, in order to find new materials or techniques to eliminate them. For superconducting qubits, for instance, parasitic two-level systems coupling to qubits or resonators have been shown to be a limiting factor for qubit coherence \cite{Lisenfeld2023}, leaving room for further hardware improvements.

On the other hand, instead of reducing the noise inherent in the hardware, one can apply algorithmic methods to handle and reduce the errors. There are primarily two approaches to algorithmic error correction \cite{Zheyu2023}, of which the first and most prominent is quantum error \textit{correction}. The core idea here is to encode logical qubits in multiple noisy physical qubits. Adaptive corrective operations are then performed based on syndrome measurements \cite{Nielsen2012}. These methods require low enough gate errors such that the qubit number overhead does not blow up. Current gate fidelities are too low for general applicability, however, promising results have been achieved by using surface codes \cite{Krinner2022, Kitaev2003, Dennis2002, Raussendorf2007}. 

One prominent source of errors not captured in error correction are readout errors, a subset of state preparation and measurement (SPAM)-errors \cite{Greenbaum2015,Geller2021}. Readout errors occur in the process of reading out the state of the qubit, e.g. spin projection measurements. Experimentally, one can understand the importance of such errors through the example of a superconducting resonator measurement of a transmon. When using only few photons to read out the resonator, there is a nonzero probability of an incorrect readout. However, by increasing the number of photons, higher qubit levels will be excited leading to a sharp drop in readout fidelity~\cite{Walter2017}.

Readout errors are better handled by the second approach, quantum error \textit{mitigation} (QEM). These kinds of methods do not correct the errors in execution of the quantum algorithm, but rather reduce the errors by post-processing the measured data. There exists a broad spectrum of QEM methods \cite{Zheyu2023,Qin2022, Endo2021}, a large class of which focuses on readout errors, such as unfolding \cite{Nachman2020,Srinivasan2022}, T-matrix inversion \cite{Geller2020}, noise model fitting \cite{Bravyi2021,Kwon2021} and detector tomography based methods \cite{Maciejewski2020}. Other promising results have recently been shown on a NISQ processor \cite{Kim2023} using zero-noise extrapolation \cite{Temme2017,Li2017}.

Our proposed method falls into the category of detector tomography-based mitigation, similar to Ref.~\cite{Maciejewski2020}, where the focus is on correction of readout probability vectors. 
Detector tomography-based methods use a set of calibration states, e.g. the eigenstates of the Pauli matrices, to characterize the measurement outcomes realized by the measurement device using the positive operator-valued measure (POVM) formalism. By looking at the reconstructed measurements, one gains information about the noise present in the experiment \cite{Chen2019}, which can then be leveraged for error mitigation. We emphasize that readout error mitigation operates on the statistics of output distributions collected from multiple runs of the experiment, and is not able to correct the results of an individual projective measurement.  

Our approach generalizes the ideas of previous protocols by using a complete characterization of a quantum system. This is done by quantum state tomography, which directly estimates the density matrix of the system \cite{Smithey1993, Liu2005}. The question of QST in error mitigation was also considered in Ref.~\cite{Maciejewski2020}, but restricted to only classical errors, i.e. errors that can be described as a stochastic redistribution of the outcome statistics of single basis measurements. Ref. \cite{Chen2019} does perform a general reconstruction of the detectors, and suggests it can be used in QST, but does not provide any protocol that goes beyond correction of marginals with numerical inversion, mostly focusing on classical errors in their applications. Similar ideas were compared to data pattern tomography in Refs.~\cite{Motka2017, Motka2021} with linear inversion. In Ref.~\cite{Ramadhani2021}, some of the error mitigation methods mentioned above were compared with simulated depolarizing noise in the context of state reconstruction. The key element that constitutes the novelty of our approach is the direct integration of QDT with QST, which leads to an error mitigation scheme that does not require matrix inversions or numerical optimizations. On top of this the protocol makes no assumptions about the error channel that represents the readout noise or about the type of quantum state (e.g.\ entangled, separable, pure, mixed) that it applies to. As we do no restrict the set of allowed POVMs in the detector tomography in any way, the protocol is also agnostic to the architecture (e.g.\ superconducting qubits, photons, ultracold-atoms, neutral atoms) and the type of readout mode (e.g.\ single projective measurement (photons), resonator readout (superconducting qubits)). 

With this general scheme, we test the limits of readout error mitigation by subjecting the protocol to common noise sources. Previous works have been extensively tested on publicly accessible IBM quantum computers, which only allow for very limited experimental control over noise sources. In this paper, we evaluate our proposed method on a chip, where we have full control over the experiment. This allows us to systematically induce and vary important noise sources, which leads to a better understanding of the strengths and limitations of the protocol.

The paper is structured as follows: in section \ref{sec:preliminaries}, we review theoretical concepts used in the protocol, such as POVMs, readout noise parametrization, quantum state tomography and a measure of reconstruction success. 
In section \ref{sec:noise_mitigated_QST} we briefly review QDT and present our readout error mitigation protocol.
Section ~\ref{sec:experimental_realization} presents the experimental setup and the implementation of the measurement protocol, followed by a discussion of the measured data resulting from the systematic variation of experimental noise sources.  
In section \ref{sec:protocol_errors}, we investigate the limitations of the protocol, with a particular focus on sample fluctuations and state preparation errors. 
A summary of the advances that the protocol offers and a review of possible future research directions is presented in section \ref{sec:conclusion_and_outlook}.

\section{Preliminaries}
\label{sec:preliminaries}

This section reviews the relevant theoretical concepts to describe general quantum measurements and state reconstruction \cite{Nielsen2012,Sakurai2017, Paris2004}. Readers familiar with these concepts may want to skip this section.

Notation used in this paper unless stated otherwise: 
Objects with tilde, e.g. $\Tilde{p}$, indicate that it has been subject to readout noise.
Objects with a hat, e.g. $\hat{p}$, are experimentally measured quantities, which contain sample fluctuations and noise. The hat is also used to denote estimators, which will be clear from the context.
Operators and objects without either accent will be considered theoretically ideal objects. 
The two eigenstates of the Pauli operators $\sigma_x, \sigma_y$ and $\sigma_z$, will be denoted by $\ket{0_i}$ and $\ket{1_i}$, where $i \in \{x, y, z\}$. If no subscript is given, the $z$-eigenstates are implied.
\subsection{Generalized measurements}
Measurements on quantum systems are described by expectation values of operators, given by Born’s rule \cite{Nielsen2012, Sakurai2017},
\begin{equation}
    \langle O \rangle=\Tr(\rho O),
\end{equation}
where $\rho$ defines the quantum system and $O$ represents a measurable quantity.
Experiments typically perform projective measurements onto a set of basis states. The simplest example is a measurement in the computational basis, with two possible outcomes. The projective operators take on the form $P_0= |0 \rangle \langle 0|$ and $P_1= |1 \rangle \langle 1|$, where $p_i=\Tr(\rho P_i )$  gives the probability of obtaining the outcome $i$.

In this work we are interested in noisy measurements. To be able to capture realistic readout we need to move to generalized quantum measurements, described by positive operator-valued measures (POVMs).  
A POVM is a particular set of operators $\{M_i\}$, sometimes called effects, that have the interpretation of yielding a probability distribution over measurement outcomes through their expectation values. In particular, any POVM has the three following properties \cite{Nielsen2012}:
\begin{equation}
    M_i^\dag = M_i, \quad  M_i\geq 0, \quad\text{and} \quad  \sum_i M_i = \mathbb{1},
\end{equation}
where the first property guarantees that the resulting expectation values are real, the second guarantees positive expectation values, and the third guarantees that the expectation values sum to unity. All together, we have a natural correspondence between each operator $M_i$ and the probability of getting outcome $i$ of a random process,
\begin{equation}
    p_i=\langle M_i \rangle=\Tr(\rho M_i).
\end{equation}
Through this general interpretation of the effects $M_i$, one can assign each effect to a possible outcome from a measurement device.

We are particularly interested in POVMs that form a complete basis in their respective Hilbert space. Such a POVM is called informationally complete (IC). Using an IC POVM allows one to decompose any operator $O$ in terms of the set $\{M_{\text{IC},i}\}$, 
\begin{equation}
        O=\sum_i c_{i} M_{\text{IC},i}.
        \label{eq:IC_POVM_decompostion}
\end{equation}
If the POVM is minimal, meaning that its elements are linearly independent, the coefficients $c_{i}$ are unique. For qubit systems to be informationally complete, we need $4^n$ linearly independent POVM effects, where $n$ is the number of qubits in the system. An example of an IC POVM is the so-called Pauli-6 POVM. Experimentally such a measurement can be performed by randomly selecting either the $x-$, $y-$ or $z-$basis and then perform a projective measurement in that basis. This POVM is defined by the set $\left\{\tfrac{1}{3}\ket{0_i}\bra{0_i}, \tfrac{1}{3}\ket{1_i}\bra{1_i}\right\}$, where $i\in \{x, y, z\}$. Note that the Pauli-6 POVM is not minimal, and its elements are not projectors, due to the prefactor acquired by the random basis selection.

\subsection{Representation of readout noise}
Measurements of quantum objects suffer from imperfections caused by the errors introduced by the measurement device. This means that when intending to perform the POVM $\{M_i\}$, we actually perform the erroneous POVM $\{\Tilde{M}_i\}$. In reading out the quantum state $\rho$ we observe the distribution dictated by the Born's rule
\begin{equation}
    \Tilde{p_i}=\Tr(\rho \Tilde{M_i}),
    \label{eq:passive_noise}
\end{equation}
which we will dub a \textit{passive} picture of noise. This runs counter to the more common view, where the noise is applied to the quantum state rather then the readout operator, hereafter referred to as the \textit{active} picture of noise, 
\begin{equation}
    \Tilde{p_i}=\Tr(\Tilde{\rho}M_i),
    \label{eq:active_noise}
\end{equation}
where $\Tilde{\rho} = \mathcal{E}(\rho)$ is the ideal quantum state passed passed through a noise channel \cite{Nielsen2012}.
These two views are equivalent, since we can only access the probabilities $\Tilde{p}_i$. 
In the remainder of this work we will be working in the passive noise picture.

\subsection{Quantum state tomography}
The objective of quantum state tomography is to reconstruct an arbitrary quantum state from only the measurement results. The most basic approach to QST can be framed as a linear inversion problem. To uniquely identify the quantum state, we use the decomposition provided by an IC POVM in eq.~\eqref{eq:IC_POVM_decompostion}, 

\begin{equation}
    \rho = \sum_i a_i M_{i},
\end{equation}
where we have dropped the subscript ``IC". 
The task of linear inversion is to determine the real coefficients $a_i$ through the recorded measurement outcomes. Linear inversion is often considered bad practice as it can yield unphysical estimates \cite{BlumeKohout2010}. It is advisable to use estimators such as the likelihood-based Bayesian mean estimator (BME) \cite{BlumeKohout2010, Gebhart2023}, or maximal likelihood estimator (MLE) \cite{Lvovsky2004}, which relies directly on the likelihood function 
\begin{equation}
    \mathcal{L}_{M}(\rho) \propto \Pi_i  \Tr(\rho M_i)^{n_i}=\left(\Pi_i \Tr(\rho M_i)^{\hat{p}_i}\right)^N,
    \label{eq:Likelihood_function}
\end{equation}
which represents the likelihood that the prepared state of the system was $\rho$, given that outcome $i$ has been observed $n_i$ times, or, equivalently with outcome frequency $\hat p_i = \tfrac{n_i}{N}$.
In our notation, the subscript $M$ indicates what POVM was used for the measurement. The estimated state is the integrated mean 
\begin{equation}
    \rho_\text{BME}\propto \int d\rho \, \rho \mathcal{L}_M(\rho)
\end{equation} for BME, and the maximizer \begin{equation}
    \rho_\text{MLE}=\underset{\rho}{\operatorname{argmax}} \, \mathcal{L}_M(\rho)
\end{equation} for MLE.

\subsection{Infidelity}
\label{infidelity_section}
Quantum infidelity is well-suited as a figure of merit for successful quantum state reconstruction. We seek to minimize the infidelity, defined as \footnote{ $\sqrt{\rho}$ is the square root of the eigenvalues of $\rho$ in an eigendecomposition $\sqrt{\rho}=V\sqrt{D}V^{-1}$, where $\sqrt{D}=\text{diag}(\sqrt{\lambda_1},\sqrt{\lambda_2}, \dots , \sqrt{\lambda_n})$.}
\begin{equation}
    I(\rho,\sigma)=1 -F(\rho,\sigma)=1-\left[\Tr(\sqrt{\sqrt{\rho}\sigma\sqrt{\rho}})\right]^2,
\end{equation}
where $F(\rho,\sigma)$ is the quantum fidelity \cite{Jozsa1994}.
It acts as a pseudo-distance measure, and is asymptotically close to the Bures distance when $1-F(\rho,\sigma)\ll 1$ \cite{Hbner1992,Uhlmann1976}.

In this paper, we only consider reconstruction of pure target states $\rho$, for which the infidelity simplifies to
\begin{equation}
    I_{\text{pure}}(\rho, \sigma) = 1 - \Tr(\rho \sigma).
\end{equation}

The infidelity between a sampled target state $\rho$ and the reconstructed state $\sigma$ is expected to decrease with  $I(\rho,\sigma) \propto N^{-\alpha}$ where $N$ is the number of shots performed and $\alpha$ is a asymptotic scaling coefficient that depends on features of the estimation problem \cite{Massar1995,Bagan2006B,Huszr2012}, such as the deviation from measuring in the target states eigenbasis, and the purity of the target state.

\section{Noise mitigated quantum state tomography by detector tomography}
\label{sec:noise_mitigated_QST}
In contrast to previous works, which view error mitigation as the correction of outcome statistics (see, e.g. Ref.~\cite{Maciejewski2020}), our approach integrates the error mitigation procedure directly into the quantum state tomography framework. This comes with crucial advantages, as we will see, because it allows for mitigation of a broader class of readout errors and intrinsically yield physical reconstruction without any additional corrective steps. We stress that this and similar protocols only work on mitigation of the statistical behaviour of the quantum system, and is not suited for mitigation on the level of single shot readout. 

\subsection{Quantum detector tomography}
\label{subsec:quantum_detector_tomography}
QDT \cite{Lundeen2008} can be summarized as reconstructing the effects associated with a given measurement outcome, based on observed outcomes from a set of calibration states. Another, perhaps more instructive, view is that QDT finds a map between the ideal and actual POVM $\{M\}\rightarrow\{\Tilde{M}\}$. Such a map could be used to extract noise parameters or learn the general behavior of the measurement device. However, our protocol avoids explicitly reconstructing such a map, and does not even require the knowledge of the ideal POVM $\{M\}$. 

QDT starts out by preparing a complete set of calibration states, which span the space of all quantum states, and repeatedly performs measurements on these states. An example is the (over)complete set of Pauli density operators, $\left\{\ket{0_i}\bra{0_i}, \ket{1_i}\bra{1_i}\right\}$, where $i \in \{ x, y, z\}$, which spans the space of single qubit operators. With the outcomes from all of these measurements, we can set up a linear set of equations given by Born's rule,
\begin{equation}
    \frac{n_{is}}{N}= p_{is}=\Tr(\rho_s M_i),
    \label{eq:QDT_set_of_equations}
\end{equation}
where the subscript $s$ iterates over the set of calibration states, e.g. $\{\rho_1=|0_x\rangle \langle 0_x| $ , $\rho_2=|1_x\rangle \langle 1_x| , \dots \}$. Eq.~\eqref{eq:QDT_set_of_equations} represents a set of $I\cross S$ constraints on $\{M_i\}$ we seek to reconstruct, where $S$ is the total number of unique calibration states prepared and $I$ is the number of POVM elements. In addition, one has the normalization constraint on the effects, $\sum_i M_i=\mathbb{1}$. Altogether, this yields a statistical estimation problem of very similar nature to the one outlined in state tomography. It can be solved, for example, by using a maximal likelihood estimator, such as the one outlined in Ref.~\cite{Fiurek2001}, which guarantees physical POVM reconstructions. 

\subsection{Readout error mitigated tomography}

\begin{figure}[t]
    \def\svgwidth{.5\textwidth}
    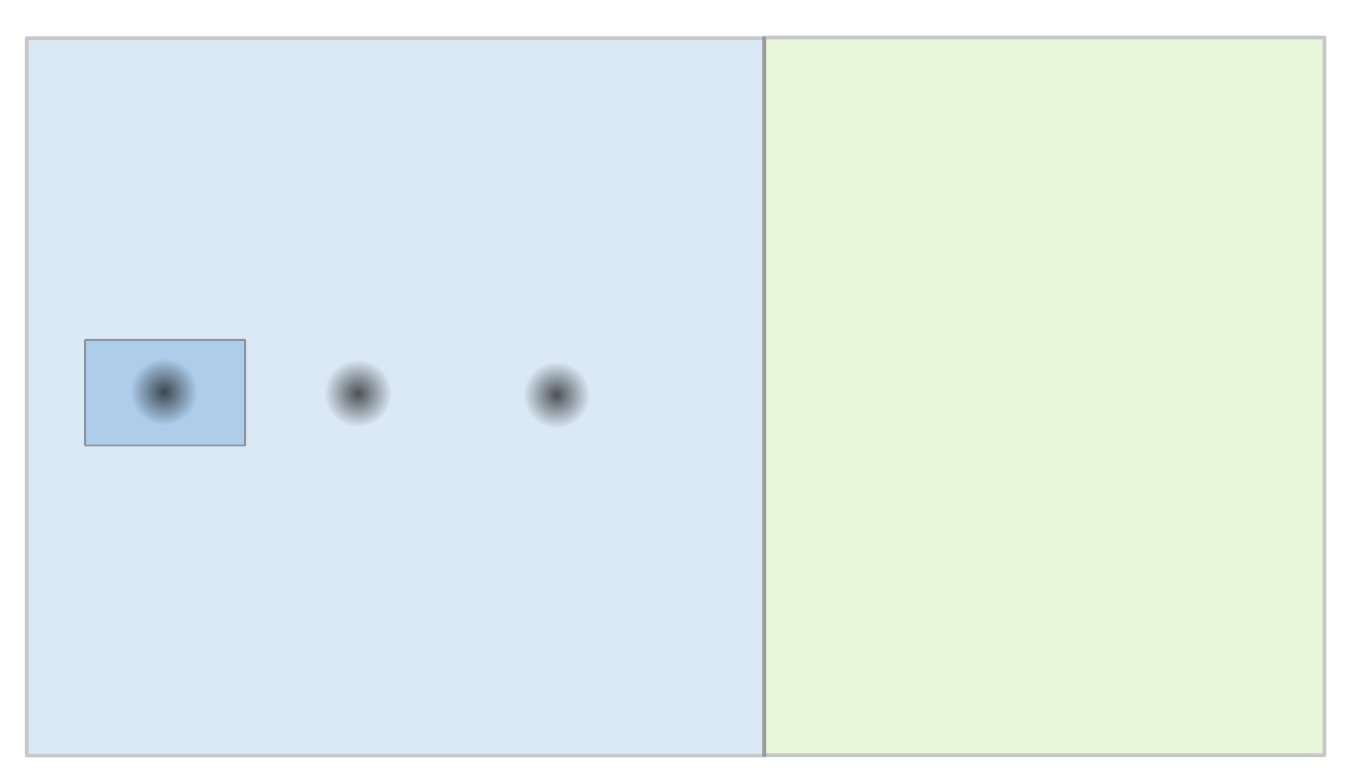
    \caption{Protocol schematic overview. \textbf{(a)} Detector tomography: A complete set of basis states (e.g. the Pauli states) are prepared and measured repeatedly by the experiment. Based on the outcomes of the measurements, a POVM is reconstructed. In a sense, it associates the measured outcome (here visualized as spin up/down measurements) with a measurement effect $\tilde M_i$. \textbf{(b)} State tomography: Using the reconstructed POVM, the modified likelihood function is endowed with knowledge of the operation of the measurement device. The system of interest is then prepared and measured repeatedly with the desired number of shots.}
    \label{fig:protocol}
\end{figure}

The core idea of our protocol is to use QDT as a calibration step before the state reconstruction. Using the information gathered from reconstructing the measurement effects, we modify the standard state estimator using the passive picture of noise, eq.~\eqref{eq:passive_noise}. In this way, the estimator is ``aware" of the noise and corresponding errors present in the measurement device.

A schematic overview of the protocol is given in Fig.~\ref{fig:protocol}. The first step is to reconstruct the POVM of the measurement device using detector tomography. This gives us access to an estimated noisy POVM $\{\tilde M^{\text{estm}}_i\}$. This noisy POVM is fed into the quantum state estimator, giving us the likelihood function
\begin{equation}
\mathcal{L}_{\tilde M^{\text{estm}}}(\rho)  \propto \Pi_i  \Tr(\rho \tilde{M}_i^{\text{estm}})^{n_i}=\left(\Pi_i (\Tilde{p}_i)^{\hat{p}_i}\right)^N,
    \label{eq:Adapted_likelihood_function}
\end{equation}
where we have used that $\tilde p_i = \Tr(\rho \tilde M_i^{\text{estm}})$. This does indeed give us an estimator that converges to the noiseless state $\rho$, see Appendix~\ref{appendix:Estimator_convergence} for more information. We highlight that working in the passive noise picture comes with benefits not enjoyed by other estimators. Firstly, no inversion of a noise channel is required, and secondly, the reconstructed state is guaranteed to be physical. Everything is handled internally by the estimator.

Our protocol assumes the following experimental capabilities:
\begin{itemize}
    \item Access to an IC POVM.
    \item Perfect state preparation.
\end{itemize}
The first capability is required by any full state reconstruction method, which can be done by any quantum device that has the ability to perform single qubit rotations and readout \footnote{By single qubit rotations we mean any operation on a quantum device that can be decomposed into tensor products of single qubit rotations acting on each qubit. Similarly for readout, one only needs tensor products of single qubit POVMs.}. This is also important for the noise mitigation to be general, since an unambiguous characterization of the erroneous state transformation inevitably requires an IC POVM. 
The second assumption is more problematic, as this is not strictly fulfilled in any real experiment. To circumvent this problem one needs to consider state preparation and readout errors in a unified framework which is an active field of research. For this work it suffices to make sure that the state preparation errors are small compared to the readout errors. We discuss this further in Sec.~\ref{sec:state-preparation-errors}. For the explicit implementation of the protocol used in the remainder of the paper, see Appendix~\ref{appendix:explicit_protocol}.

\section{Experimental realization}
\label{sec:experimental_realization}

We implemented the protocol introduced in the previous section on a fixed-frequency transmon qubit \cite{Koch2007, Kjaergaard2020} with frequency $\omega_{01} = 6.3\, \text{GHz} $ coupled to a resonator with frequency $\omega_r = 8.5\,\text{GHz}$ in the dispersive readout scheme.
Typical coherence times were observed to be $T_1 = 30\pm5\,\text{$\mu$s}$ and $T_2 = 28\pm6\,\text{$\mu$s}$.
For more information about different types of superconducting qubits, in particular their operation and noise sources affecting them, we refer to Ref.~\cite{Krantz2019}.

\subsection{Experimental setup}
A schematic diagram of the setup is shown in Fig.~\ref{fig:cryo_setup}.
All measurements without an explicitly stated temperature were performed at $10\,\text{mK}$, using a dry dilution refrigerator. 
At room temperature we use a time-domain setup with superheterodyne mixing of readout pulses triggered by a high-frequency arbitrary wave generator (AWG) responsible for local oscillator readout and pulsed manipulation tones. Signal lines are attenuated by $70 \, \text{dB}$ distributed over various stages, as seen in Fig.~\ref{fig:cryo_setup}.
The measured signals are converted into complex numbers in the so-called IQ plane \footnote{IQ abbreviates in-phase and quadrature, the labels used for the two signals recorded for each measurement. These two signals can be transformed equivalently into amplitude and phase of the signal. By doing many measurements on the system, two distinct distributions become visible, each corresponding to the collapsed state post-measurement: $\ket{0}$ and $\ket{1}$}. With optimized experimental calibrations and no induced noises, we can distinguish between two non-overlapping Gaussian distributions corresponding to the two possible outcomes of the measurement, with the qubit state being projected to either ground state $\ket{0}$ or excited state $\ket{1}$. For examples of how these distributions change in the presence of noise, see Fig.~\ref{fig:IQ-scatter} in Appendix~\ref{app:experimental calibrations}.
Unless otherwise stated, we perform measurements at the best separation settings for the classification model. This usually corresponds to approx. $98\%$ distinguishability, meaning that $98\%$ of shots can be correctly categorized as spin up or spin down.

\begin{figure}[t]
    \includegraphics[width=.9\linewidth]{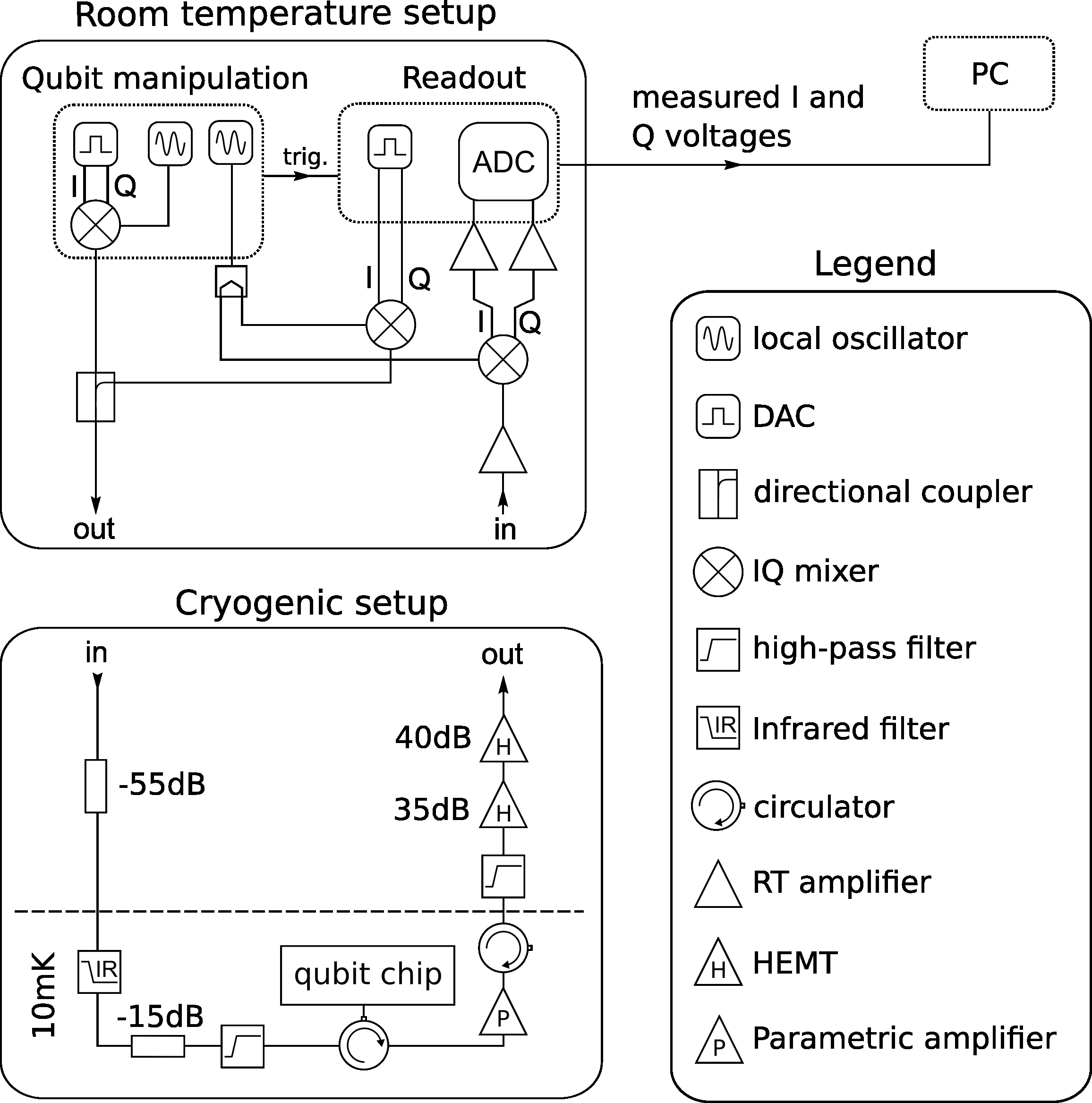}
    
    \caption{\justifying Microwave measurement setup. The local oscillator (LO) and phase-shifted pulsed IQ signals are mixed to provide the readout signal for the dispersive measurement of the qubit. A separate $\omega_{01}$ pulsed signal manipulates the qubit. The attenuators on the mixing chamber stage are carefully thermalized to $\text{mK}$ temperatures to mitigate extra thermal noise. A wideband Josephson travelling-wave parametric amplifier (TWPA) and the two HEMT-based amplifiers provide quantum-limited readout.}
    \label{fig:cryo_setup}
\end{figure}

\subsection{Measurement protocol}
\label{sec:measurement_protocol}

For each noise source we study, we set up and execute our experiment as follows:

\begin{enumerate}
    \item Estimation of $\pi$-pulse length ($T_\pi$) by Rabi-oscillations close to the qubit frequency.
    \item Refined measurement of $\omega_{01}$ by a Ramsey-experiment.
    \item Measurement of $T_{\pi}$ from a Rabi-measurement with updated $\omega_{01}$ from 2.\ by fitting a decaying sine function to the data.
    \item Ground, $\ket{0}$, and excited states, $\ket{1}$,  are measured with a single-shot readout. The location of the states in the IQ plane is learned by a supervised classification algorithm. For more information, see Appendix~\ref{app:experimental calibrations}.
    \item QDT is performed by preparing each of the six Pauli states, and measuring them in the three bases $\sigma_x, \sigma_y$ and $\sigma_z$, which is described by the Pauli-6 POVM. The first two of these measurements are done as a combination of qubit rotation and subsequent $\sigma_z$ readout. Quantum states are prepared using virtual Z-gates.
    \item QST is performed and averaged over 25 random quantum states $U \ket{0}$, where $U$ is a random unitary (Haar-random). Each state is measured in the three Pauli bases. 
\end{enumerate}
Performing this experiment for various strengths of the noise allows us to benchmark the ability of the protocol to mitigate the given noise source.
For a schematic diagram of the measurement pipelines, see Fig.~\ref{fig:Experimental_flow}.

The outcome of two example experimental runs are shown in Fig.~\ref{Fig: QST-improved}. Both standard and mitigated QST infidelities are extracted on a shot-by-shot basis. Typical features include a priori infidelity of roughly $0.5$, with a rapid power-law decay, saturating at a given infidelity level for unmitigated QST. After this saturation is reached, further measurements will not improve the quantum state estimate because the measurements performed on the system are noisy. By using quantum readout error mitigation (labeled as QEM), we can significantly lower infidelity, enabling more precise state reconstruction. For all averaged experiments, the mitigated QST infidelities are consistently below the unmitigated QST infidelities.

\begin{figure}[t]
    \includegraphics[width=.75\linewidth]{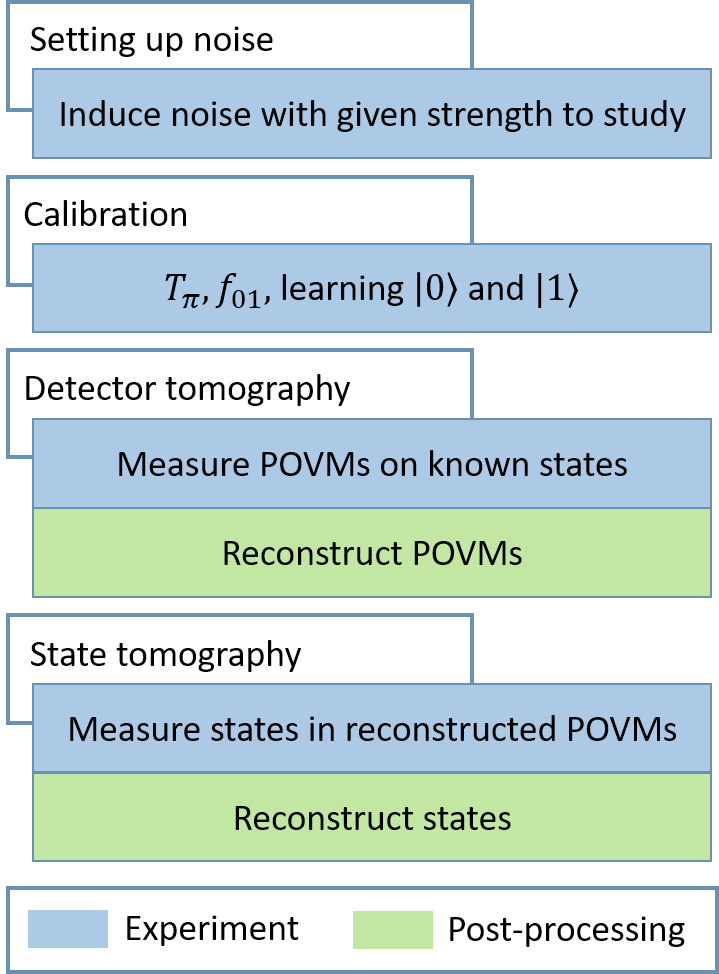}
    \caption{ \justifying Schematic of the experimental pipeline of the protocol. In the first step, a given noise is induced with a specific strength. The experimental readout is calibrated for this noise. Detector tomography is performed, which reconstructs the noisy Pauli POVM $\{\Tilde{M_i}\}$. Finally, quantum state tomography is executed and reconstruction infidelity is evaluated and averaged over a set of randomly chosen target states. } 
    \label{fig:Experimental_flow}
\end{figure}

\begin{figure*}[t]
    \includegraphics[width=\linewidth]{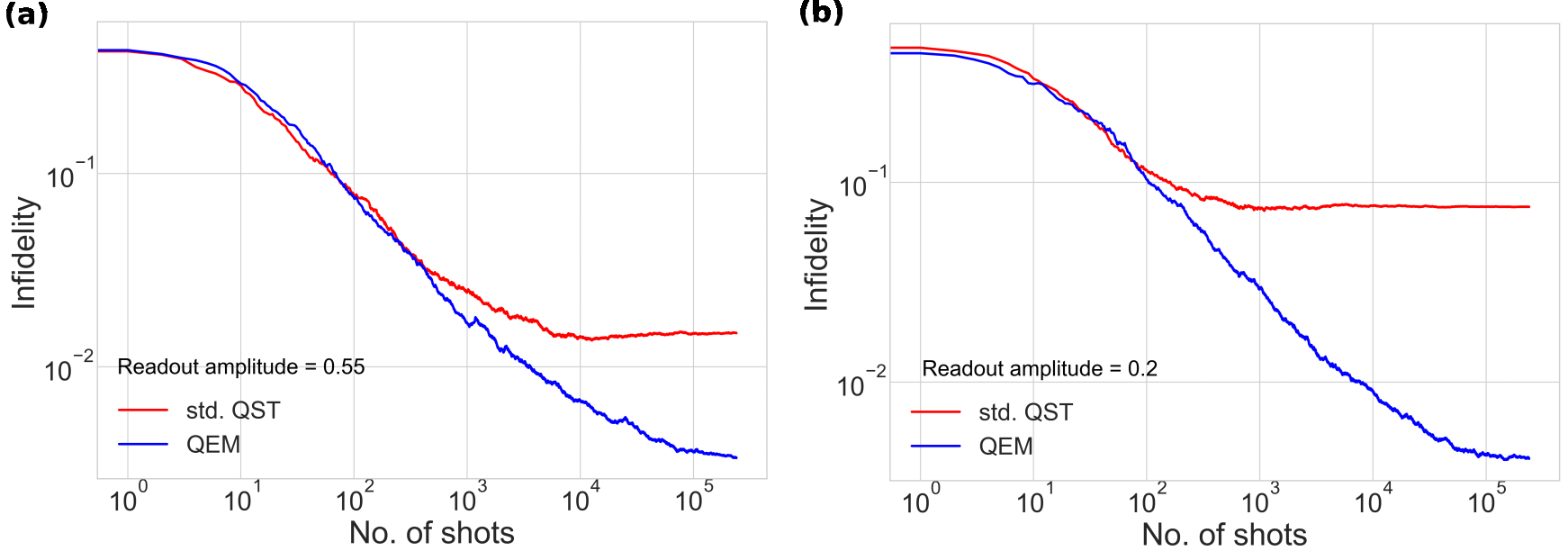}
     \caption{Mean quantum infidelity of reconstructed quantum states as a function of number of shots for \textbf{(a)} optimal readout powers and \textbf{(b)} weak readout. The red lines represent unmitigated, standard QST, blue lines are obtained with error-mitigated QST. Error mitigated QST reaches a lower infidelity saturation value than standard reconstruction in both cases. At low readout power, error mitigated QST can still reconstruct the state with similar accuracy, while standard QST saturates at a significantly higher value.}
    \label{Fig: QST-improved}
\end{figure*}

\begin{figure*}[t]
    \includegraphics[width=\linewidth]{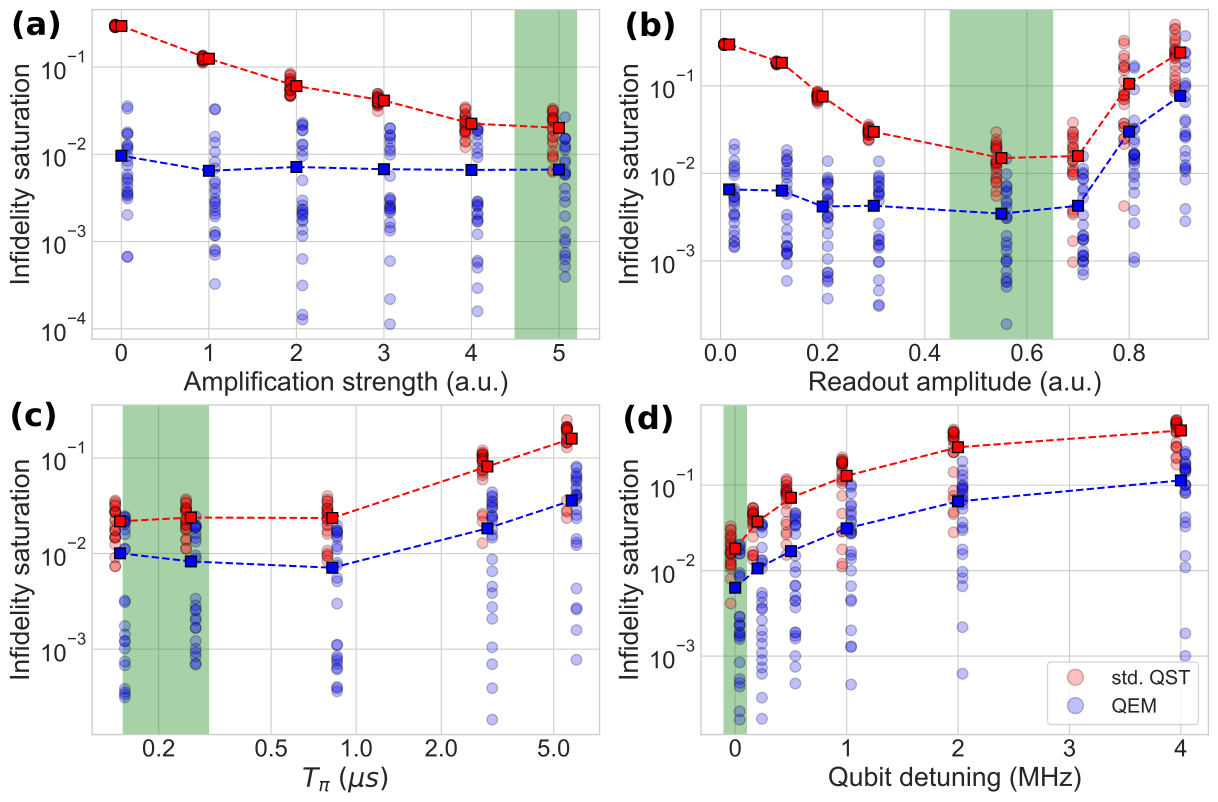}
    \caption{ Characterizing error mitigation by QDT for QST for different noise sources. Infidelity saturation refers to the last infidelity point measured over 240k single-shot measurements. The infidelity saturation for each individual run is plotted in translucent circles and shifted off center, to the left for standard QST and to the right for QEM, for better visibility. The solid colored squares are the average infidelity saturation, connected by dotted lines for guidance. The green highlighted areas indicate the optimal experimental parameters.  \textbf{(a)} Decreasing parametric amplification has a significant effect on QST through lesser distinguishability of the two states. Such errors can be mitigated very effectively by detector tomography. Zero amplification strength corresponds to having turned off the amplifier. \textbf{(b)} An incorrectly set readout amplitude of the resonator leads to increased infidelity in both the too weak and too strong readout regimes. Mitigation fails at higher powers, because higher levels are excited that are not taken into account. \textbf{(c)} Increasing manipulation timescales leads to more $T_1$ and $T_2$ decay events. This can be mitigated to some extent by the protocol, as seen by the smaller gradient of the mitigated infidelity curve. \textbf{(d)} Detuning between the qubit transition and drive frequencies can be efficiently mitigated. The mitigated infidelities rise proportionally to the unmitigated ones.
    }
    \label{Fig: fidelity-gains-4}
\end{figure*}

\subsection{Error sources}

 To probe the generality and reliability of our error-mitigation protocol, we artificially induce a set of noise sources, which introduce readout errors. In particular, we study:
\begin{enumerate}

    \item Errors introduced by insufficient readout amplification: study through variation of the amplification of the parametric amplifier.
    \item Errors introduced by low resonator photon number: study through variation of photon population in resonator by a variation of the readout amplitude.
\end{enumerate}
In addition, we investigate two noise sources which manifest not only as as readout errors, but also as state preparation errors. We discuss the implications of state preparation errors in more detail in Sec.~\ref{sec:state-preparation-errors}:
\begin{enumerate}
\setcounter{enumi}{2}
    \item Errors introduced by energy ($T_1$) and phase relaxation ($T_2$): study through variation of the drive amplitude.
    \item Errors introduced by qubit detuning: study through the variation of the applied manipulation pulse frequency.
\end{enumerate}
Lastly, we combine multiple error sources to simulate a very noisy experiment, where conventional reconstruction methods fail.

Since the infidelity scaling is state-dependent, each experiment is averaged over 25 Haar-random states to get a reliable average performance (see Appendix~\ref{appendix:explicit_protocol} for more information). For QDT, we perform a total of $6\cross 3 \cross80000$ (6 states, 3 measurement bases) shots, for the following QST $25 \cross 3 \cross 80000$ (25 states, 3 measurement bases, 80000 shots each).

\subsection{Experimental results}
\label{sec:Results}

\subsubsection{Insufficient readout amplification}
We study the influence of insufficient signal amplification by the wide band travelling wave parametric amplifier, in the following called amplifier. All other amplifiers operate with their optimal operational settings.
We tune the amplification through the power of microwave drive of the amplifier. This results in a reduced distinguishability of the two states. The measured saturation values of infidelity for different amplifications is given in Fig.~\ref{Fig: fidelity-gains-4}~\textbf{(a)}. Lowering the amplification reduces distinguishability in the IQ plane, resulting in high infidelities of around 0.3 with the amplifier turned off. The readout error mitigation protocol performs very well against this type of noise. When error mitigation is used, the reconstruction infidelity does not increase unless the amplifier is turned off completely, indicating that readout error mitigation is effective against this type of noise.

\subsubsection{Readout resonator photon number}

Ideally, superconducting qubits are read out with low resonator populations. Higher power readout will excite the qubit, potentially also to higher states outside the qubit manifold. Low readout power also enables non-demolition experiments, e.g. for active reset. If the readout power is too low, the IQ-plane state separation is not sufficient for single-shot readout, increasing the infidelity, an effect that our protocol can mitigate, as can be seen in Fig.~\ref{Fig: fidelity-gains-4}~\textbf{(b)}. This effectively constitutes a lower signal-to-noise ratio, similar to a low parametric amplification. At higher readout power the higher transmon states excited by the strong readout are not captured in the two-level quantum system simulation, resulting in increased infidelities for both mitigated and unmitigated reconstruction. We can also see a higher spread of unmitigated infidelities when using readout amplitudes above and including 0.55. A possible explanation for this lies in the fact that the $0\rightarrow2$, $0\rightarrow 1 \rightarrow 2$ and $1\rightarrow2$ processes have different frequencies, hence these transitions will be induced spontaneously by readout signals with a probability dependent on the current quantum state. This results in a state-dependence of higher state excitations, and therefore increased infidelity. We note that the infidelity ratio between standard QST and QEM becomes constant at large readout amplitudes.

In the following, we study noise sources that do not exclusively manifest as readout errors, as they also introduce significant errors in state preparation, which breaks one of the core assumptions of the protocol. It is nonetheless insightful to investigate the protocol's performance with such noise sources. Further experiments are required to determine the efficacy of the protocol under such conditions.

\subsubsection{Shorter \texorpdfstring{$T_1$}{T1} and \texorpdfstring{$T_2$}{T2} times}
Decoherence is modeled by exponential decay in fidelity in the $z$ ($T_1$) and azimuthal ($T_2$) directions of the Bloch sphere. Thus, we can increase the number of decay events by increasing the length of manipulation pulses, used for state preparation and readout rotations. Experimentally, this is done by decreasing the manipulation power.

Fig.~\ref{Fig: fidelity-gains-4}~\textbf{(c)} shows the results for varying the $\pi$-pulse lengths.
At increased manipulation lengths, the qubit has more time to decay both through dephasing and energy loss, resulting in a larger error probability. Both unmitigated and mitigated state reconstruction gets progressively less accurate with increased manipulation lengths, with a noise-strength independent factor of 3-4 between them. 
The combined effect of decoherence on all three basis measurements is expected to be highly state-dependent, which can explain the large infidelity saturation spreads at a given noise strength.

Longer readout pulses lead to state decay already in the state preparation stages, hence the studied noise source is not only a readout error, but falls into the broader SPAM error category. One could restrict induced decay to purely affect readout, e.g. by only decreasing the manipulation power at the readout stage, but this would arguably not correspond to a realistic experimental scenario. Instead, consistently decreasing the manipulation power presents a more relevant experimental scenario.

\subsubsection{Qubit detuning}

If we consistently apply pulses detuned in frequency, we observe an apparent infidelity reduction by a factor of 4 by the protocol across all detunings, as shown in Fig.~\ref{Fig: fidelity-gains-4}~\textbf{(d)}. One can also see, that even a small detuning of $0.1 \,\text{MHz}$ results in lower infidelity than standard QST with perfect frequency-matching. 
We note that the infidelity saturation of QEM grows immediately with detuning, meaning that the protocol is not able to completely mitigate even small detunings.
This can be attributed to the detector tomography step also suffering severely from effects of state preparation inaccuracies. 
As in the case of increased manipulation times, one could conduct an experiment where the detuning only affects the readout stage by applying resonant pulses for state preparation. As this also does not correspond to a realistic experimental scenario, the more relevant approach is to apply consistently detuned pulses.

At $4 \, \text{MHz}$ detuning a pulse of $200 \, \text{ns}$ will result in an accumulated phase offset of $0.8$, effectively scrambling $\sigma_x$ and $\sigma_y$ measurement outcomes.
Both mitigated and conventional QST infidelity saturation depend roughly linearly on detuning, flattening out when approaching an infidelity of 0.5.
The large spreads of the measured infidelity bounds stem from the fact that the effect of the noise depends on the state.

Noise that effectively manifests as a large frequency detuning can be induced, for example, by spontaneous iSWAP ``operations" with neighboring two-level-systems (TLS). Microwave sources may also suffer from frequency drifts. The errors coming from microwave source offsets are much smaller than the considered detunings. However, TLS-enabled jumps of more than $10 \, \text{MHz}$ have been observed on the same timescale as one run of QST takes for our experiment \cite{Meiner2018}, making it a potential noise source corresponding to larger qubit detunings. 
It is important that these jumps do not occur after the device tomography stage, otherwise the reconstructed POVMs will not reflect the jump.

We remark that for all the shown cases where readout error mitigation does not work optimally, specifically in \textbf{(c)}, \textbf{(d)} and the right part of \textbf{(b)} in Fig~\ref{Fig: fidelity-gains-4}, the ratio between standard QST and QEM infidelities is approximately constant. 
For noise sources that do not exclusively affect the readout process a deterioration of the noise mitigated reconstruction fidelity is expected as the noise strength increases. However, the observation of proportionality between the mitigated and unmitigated infidelities requires further investigation.

\subsubsection{Combining noise sources}

Several of the previously studied noise sources can be combined in different ways to simulate a noisy experiment. To demonstrate the efficiency of the mitigation ad extremum, we devise an experiment with comparably high levels of noise, by measuring with off-resonant drive of $\delta \omega = 0.5 \, \text{MHz}$, decreased parametric amplification and with decreased readout amplitude. While conventional QST saturates very early in the reconstruction, few-percent infidelities are possible by employing error mitigation, thereby opening the way for precise measurements with very noisy readouts, see Fig.~\ref{fig:mixed_noises}. As expected (see Sec.~\ref{infidelity_section}), one can observe a linear dependence of infidelity on shot number on a log-log scale, corresponding to a power-law convergence.

\begin{figure}[t]
    \includegraphics[width=\linewidth]{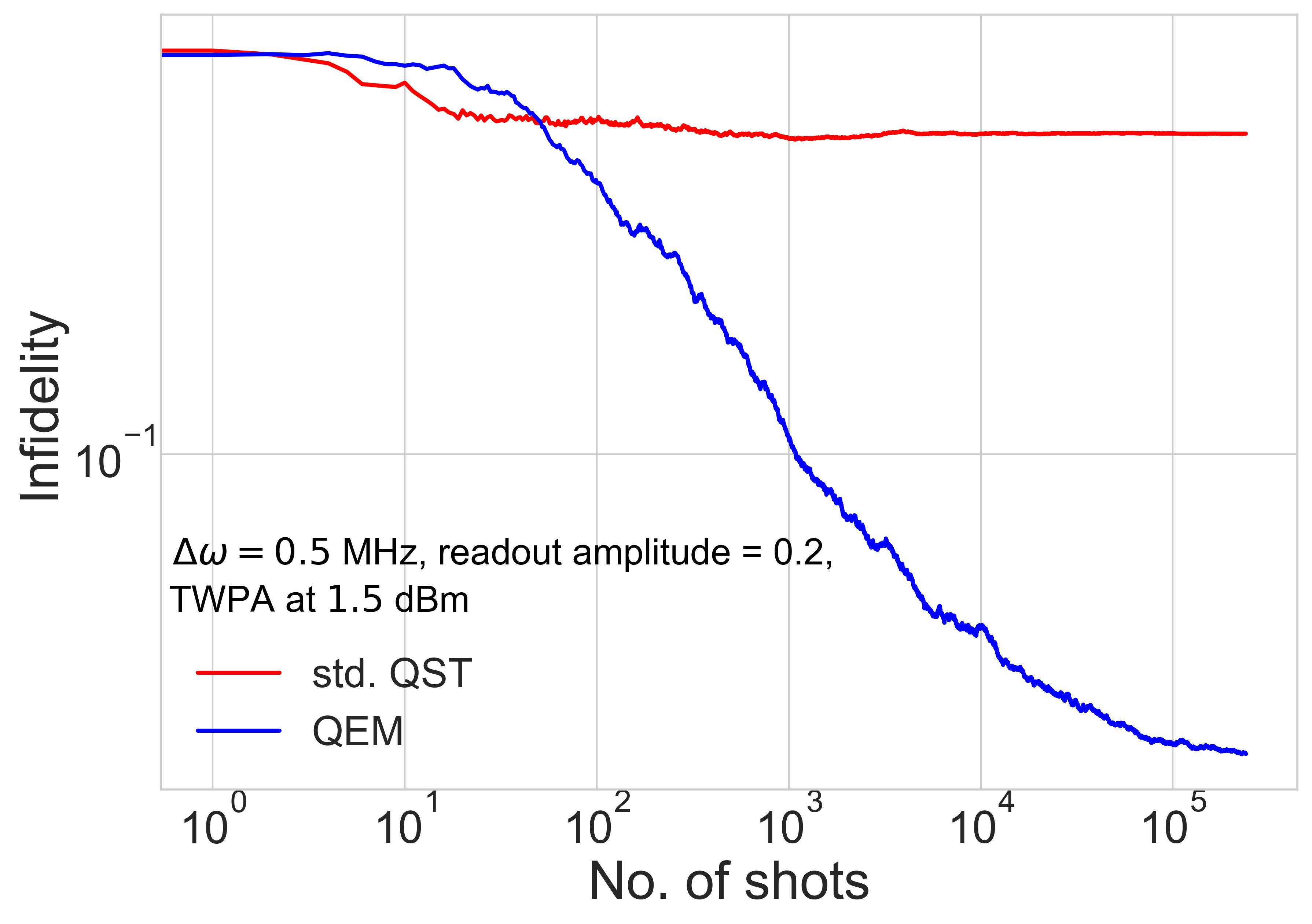}
    \caption{\justifying 
      Combination of multiple strong noise sources: Conventional QST saturates at around 100 shots, while using QDT enables state reconstruction with a factor of 30 lower infidelity. }
    \label{fig:mixed_noises}
\end{figure}

\section{Protocol limitations}
\label{sec:protocol_errors}
We present an analysis of the potential limitations to the performance and reliability of the error mitigation protocol. Note that most of the discussed limitations will affect any estimator, and are not unique to our protocol. It is still important to understand their impact on protocol performance.

\subsection{Sample fluctuations}
\label{sec:sample_fluctuations}
A central limiting factor to the precision of any estimator are sample fluctuations, i.e. statistical fluctuations in the number of samples per effect. These fluctuations can manifest themselves in two parts of the protocol, QDT and QST, and cannot be mitigated.

\begin{figure}[t]
    \includegraphics[width=\linewidth]{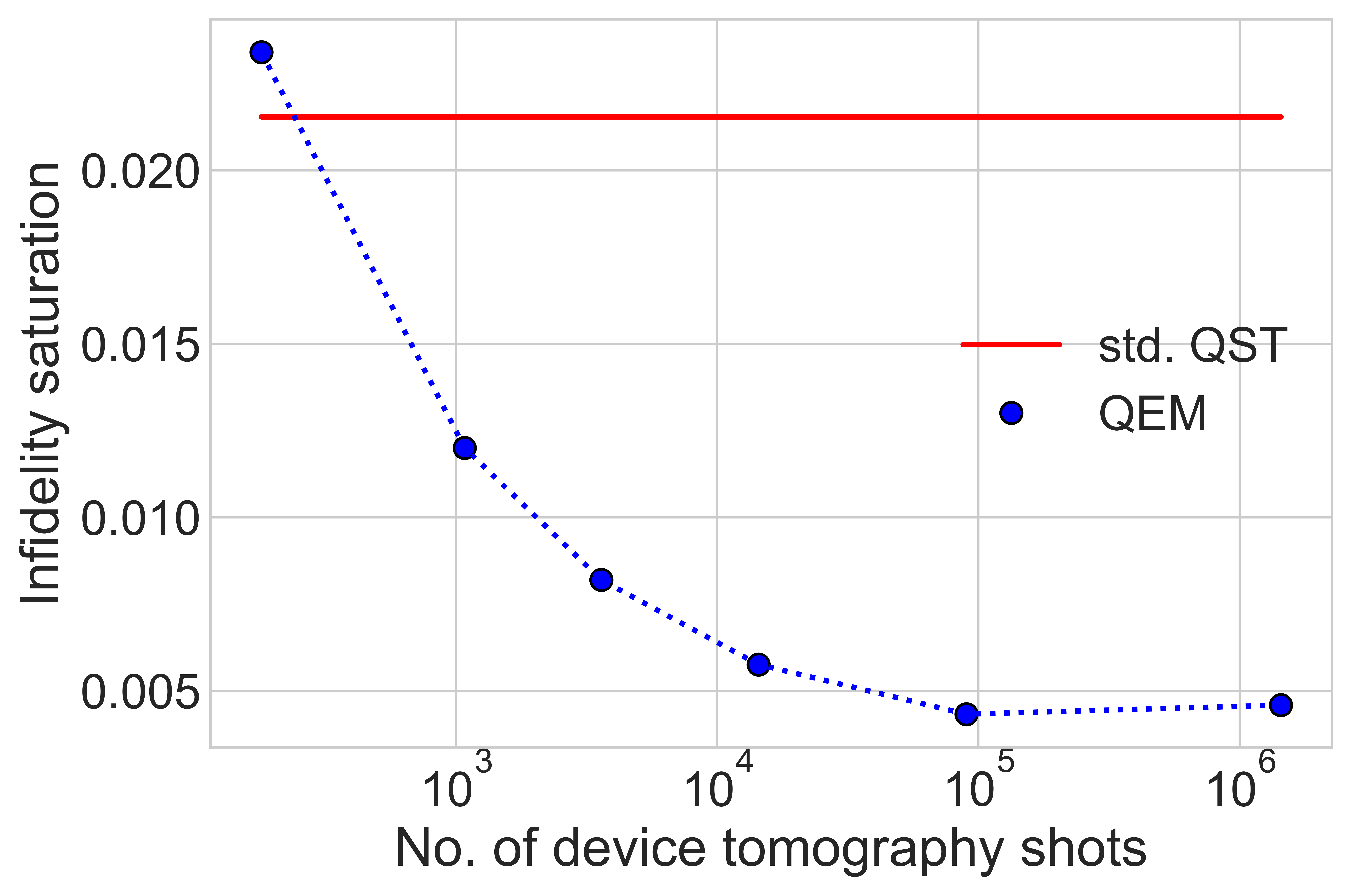}
    \caption{\justifying 
     Averaged infidelity saturation in QST as a function of calibration shots used in QDT. 240k shots were used for each QST reconstruction, averaged over 25 Haar-random states. Using 1k shots for calibration already results in a decrease of infidelity by a factor of $\approx 2$. The curve flattens out at around 100k calibration shots.}
    \label{fig:DT_dependence}
\end{figure}

QDT acts as a calibration step and is a one-time cost in terms of samples. Fluctuations in the POVM reconstruction can be viewed as a bias introduced into the state reconstruction. Therefore, it is important that the POVM reconstruction does not impose a bias larger than the expected sample fluctuation in the QST itself. In Fig.~\ref{fig:DT_dependence} we demonstrate that using less than $0.5\%$ of experimental shots for detector tomography lowers the reconstruction infidelity to half its value. The lowest infidelity is achieved by using ca. $10\%$ of shots for QDT, after which no meaningful further improvement was seen in our experiment. We give the recommendation of using half the number of shots used for a single QST for QDT. When averaging over e.g. 25 quantum states for a representative benchmark, this becomes a small overhead. 

\begin{figure}[t]
    \def\svgwidth{.45\textwidth}
    {\large \textbf{
\begingroup%
  \makeatletter%
  \providecommand\color[2][]{%
    \errmessage{(Inkscape) Color is used for the text in Inkscape, but the package 'color.sty' is not loaded}%
    \renewcommand\color[2][]{}%
  }%
  \providecommand\transparent[1]{%
    \errmessage{(Inkscape) Transparency is used (non-zero) for the text in Inkscape, but the package 'transparent.sty' is not loaded}%
    \renewcommand\transparent[1]{}%
  }%
  \providecommand\rotatebox[2]{#2}%
  \newcommand*\fsize{\dimexpr\f@size pt\relax}%
  \newcommand*\lineheight[1]{\fontsize{\fsize}{#1\fsize}\selectfont}%
  \ifx\svgwidth\undefined%
    \setlength{\unitlength}{595.27559055bp}%
    \ifx\svgscale\undefined%
      \relax%
    \else%
      \setlength{\unitlength}{\unitlength * \real{\svgscale}}%
    \fi%
  \else%
    \setlength{\unitlength}{\svgwidth}%
  \fi%
  \global\let\svgwidth\undefined%
  \global\let\svgscale\undefined%
  \makeatother%
  \begin{picture}(1,0.96666667)%
    \lineheight{1}%
    \setlength\tabcolsep{0pt}%
    \put(-0.63495302,2.06349255){\color[rgb]{0,0,0}\makebox(0,0)[lt]{\begin{minipage}{5.42596196\unitlength}\raggedright \end{minipage}}}%
    \put(-1.26990603,2.93354202){\color[rgb]{0,0,0}\makebox(0,0)[lt]{\begin{minipage}{6.45535583\unitlength}\raggedright \end{minipage}}}%
    \put(0,0){\includegraphics[width=\unitlength,page=1]{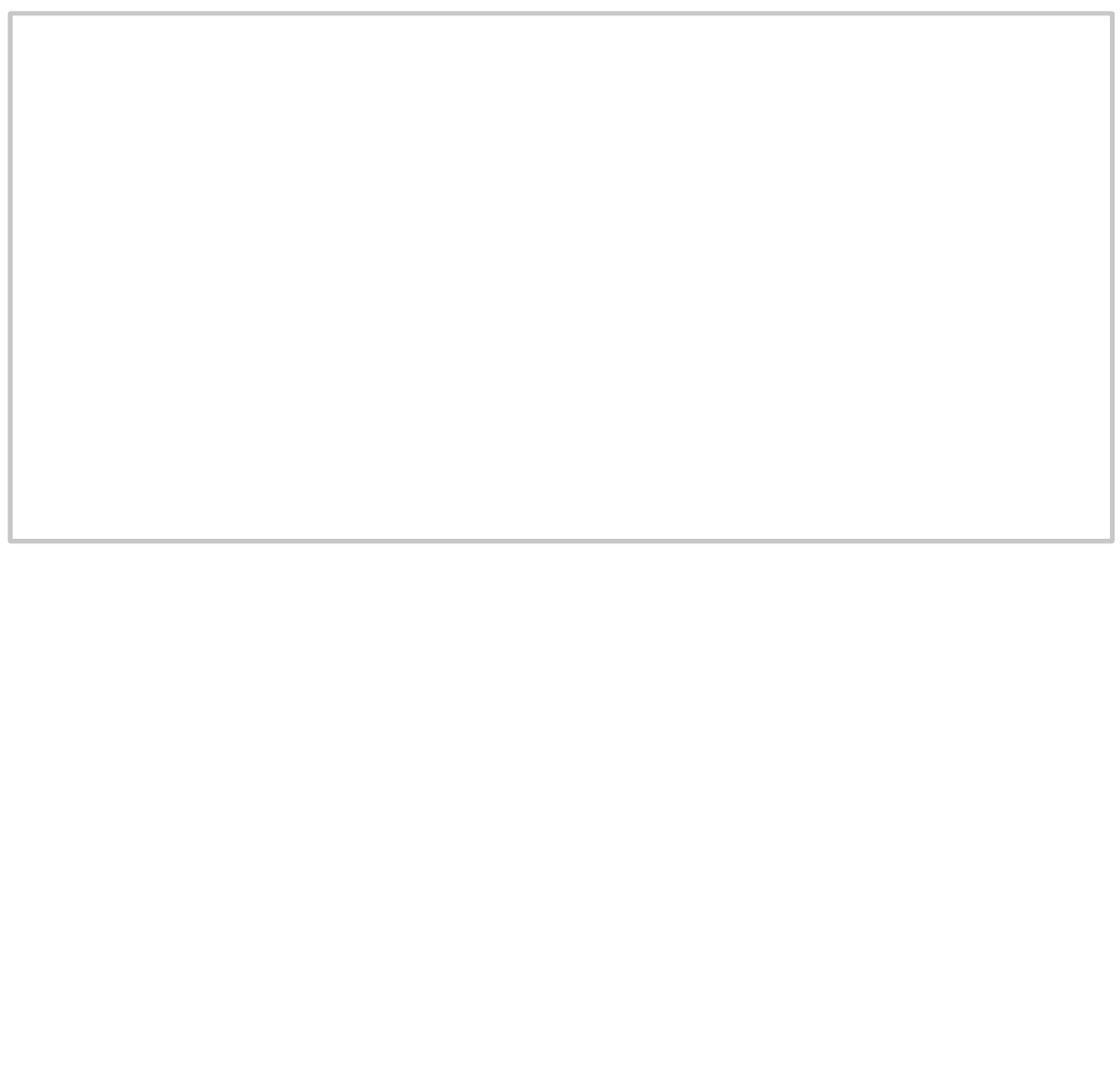}}%
    \put(0.03809734,0.89392406){\color[rgb]{0,0,0}\makebox(0,0)[lt]{\lineheight{1.25}\smash{\begin{tabular}[t]{l}(a) Active\end{tabular}}}}%
    \put(0.03809734,0.41586144){\color[rgb]{0,0,0}\makebox(0,0)[lt]{\lineheight{1.25}\smash{\begin{tabular}[t]{l}(b) Passive\end{tabular}}}}%
    \put(0,0){\includegraphics[width=\unitlength,page=2]{bias-variance.pdf}}%
  \end{picture}%
\endgroup%
}}
    \caption{\justifying  Sketch of the likelihood function on the Bloch disk. The Bloch disk represents an intersecting plane through the center of the Bloch sphere, e.g. the x-y-plane of the Bloch sphere. The red ribbon indicates the boundary of the Bloch disk, whereas the green ribbon indicates the boundary of the Bloch disk under the effect of a depolarizing channel. The green point is the true prepared state, which, due to the depolarizing readout noise, produces noisy measurement data. The red and blue points are arbitrary reference states in their active/passive noise representations. \textbf{(a)} An unmitigated likelihood function in the active representation of the noise, (see eq.~\eqref{eq:active_noise}).  \textbf{(b)}  An error mitigated likelihood function in the passive representation of noise (see eq.~\eqref{eq:passive_noise}).
    }
    \label{fig:bias-variance}
\end{figure}

\begin{figure}[t]
    \includegraphics[width=\linewidth]{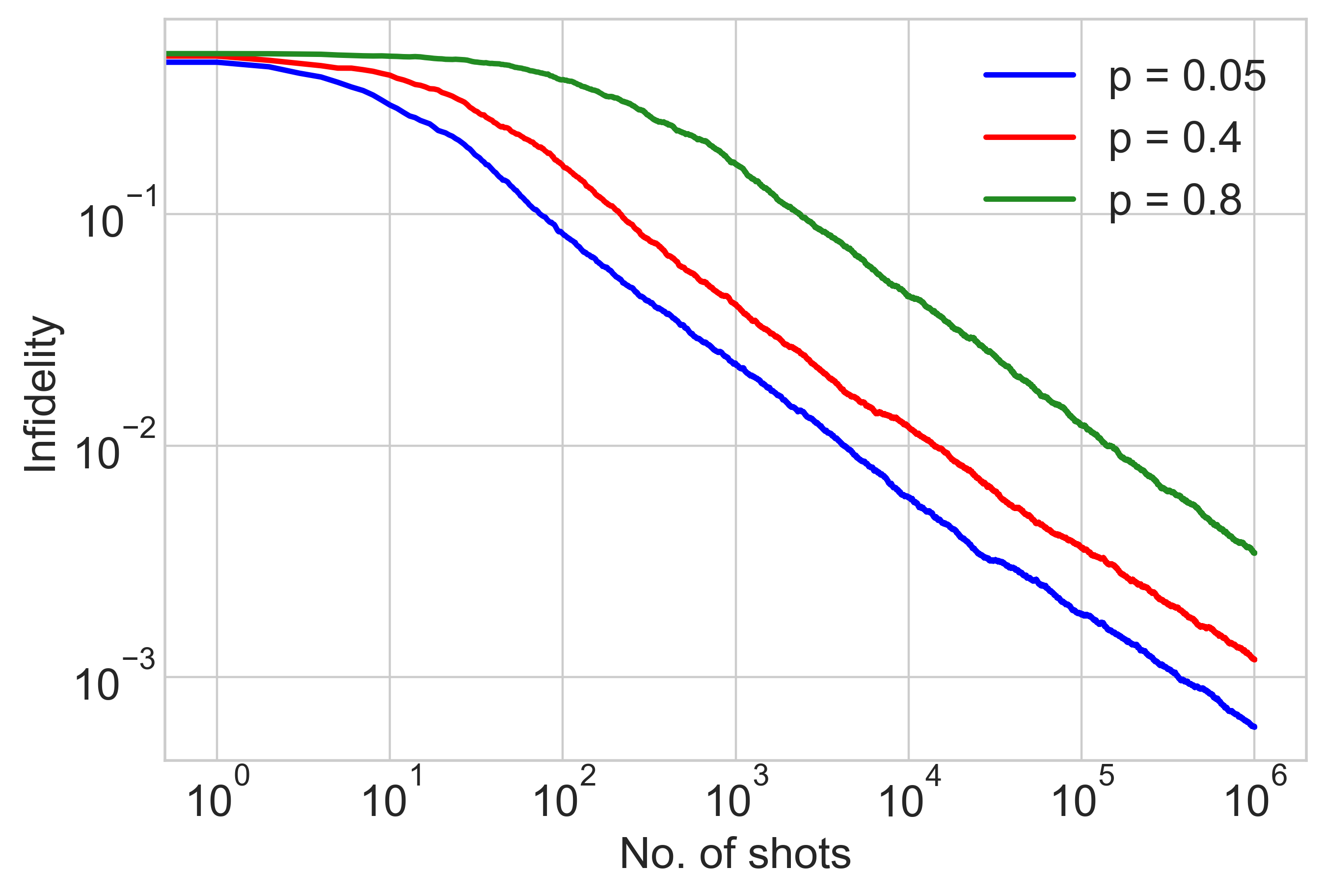}
    \caption{\justifying Reconstruction infidelity for mitigated simulated depolarizing noise for different noise strengths $p$. The infidelity is averaged over 100 Haar-random pure states. The depolarizing channel is given in eq.~\eqref{eq:depolarizing_channel}.}
    \label{fig:bias-variance_depol}
\end{figure}

Error mitigated estimators, such as the one in eq.~\eqref{eq:Adapted_likelihood_function}, aim to reduce the bias with respect to the prepared state caused by a noisy readout process. In reducing the bias we are in effect taking into account that the state has passed through a noise channel. Since noise channels reduce the distinguishability between states, the variance of the modified likelihood function is increased with respect to the unmodified likelihood function, an effect known as the \textit{bias-variance trade-off} \cite{Cai2022,Takagi2022}. In Fig.~\ref{fig:bias-variance} we give an artistic rendition of the bias-variance trade-off under mitigation of single qubit depolarizing noise, given by the channel
\begin{equation}
    \mathcal{E}(\rho)= \frac{p}{2}\mathbb{1}  + (1-p)\rho.
    \label{eq:depolarizing_channel}
\end{equation}
We compare \textbf{(a)} an unmitigated estimator in the active noise picture, to \textbf{(b)} the error mitigated estimator in the passive noise picture. In the error mitigated estimator the bias with respect to the true state (green point) is removed at a cost of higher variance of the likelihood function. We note that the bias due to the statistical fluctuations in the estimator remains, making it consistent with always-physical state estimators \cite{Schwemmer2015}.

The increased variance can be combated by increasing the number of shots used for estimation. The additional number of shots required for the error mitigated estimator to obtain the same confidence (variance of the likelihood) with respect to the true state as the unmitigated estimator with respect to the biased state is called the \textit{sampling overhead}.
The overhead relates to the distinguishability of quantum states and the data-processing inequality \cite{Takagi2022}. It tells us that the stronger the noise-induced distortion is, the larger the overhead. 
In QST, the sample overhead manifests itself as a shift in the infidelity curve. In Fig.~\ref{fig:bias-variance_depol} we have simulated varying depolarizing strengths and applied error mitigation. We see a clear shift as the noise strength increases. Note that such a shift would be present for any error mitigation schemes, including an inversion of the noise channel.

\subsection{State preparation errors}
\label{sec:state-preparation-errors}
A generic problem of QDT based methods is that they assume perfect preparation of calibration states, which is impossible in experiments. A common argument when using QDT is that the state preparation has a small error compared to the readout error itself. There are, however, methods developed to deal with the potential systematic error introduced by QDT, see e.g. Refs.~\cite{Gebhart2023, Zhang2020}. 

Another approach is to use the fact that preparation of the calibration states used in QDT only requires single qubit gates. If one has access to error estimates from single qubit gates, one could in principle also correct for state preparation errors in a similar manner to what is done in the protocol in this work (see e.g. randomized benchmarking \cite{Knill2008} for gate characterization).  

Despite our protocol requiring perfect state preparation, we have investigated noise sources that also affect state preparation, such as qubit detuning and increased $T_1$ and $T_2$ manipulation times. We will add to this investigation an experiment performed at higher qubit temperatures. We systematically increase the temperature of $10\, \text{mK}$ up to over $200 \, \text{mK}$, where the state in thermal equilibrium will have non-negligible contributions from the excited state. For example, a qubit temperature of $40\, \text{mK}$ corresponds to an excited state population of $0.05\, \%$,  and a temperature of $120\, \text{mK}$ corresponds to $7.3\, \%$  excited state population. Hence, for our experiment without active feedback, it is not possible to reliably prepare a pure calibration state. 

\begin{figure}[t]
    \includegraphics[width=\linewidth]{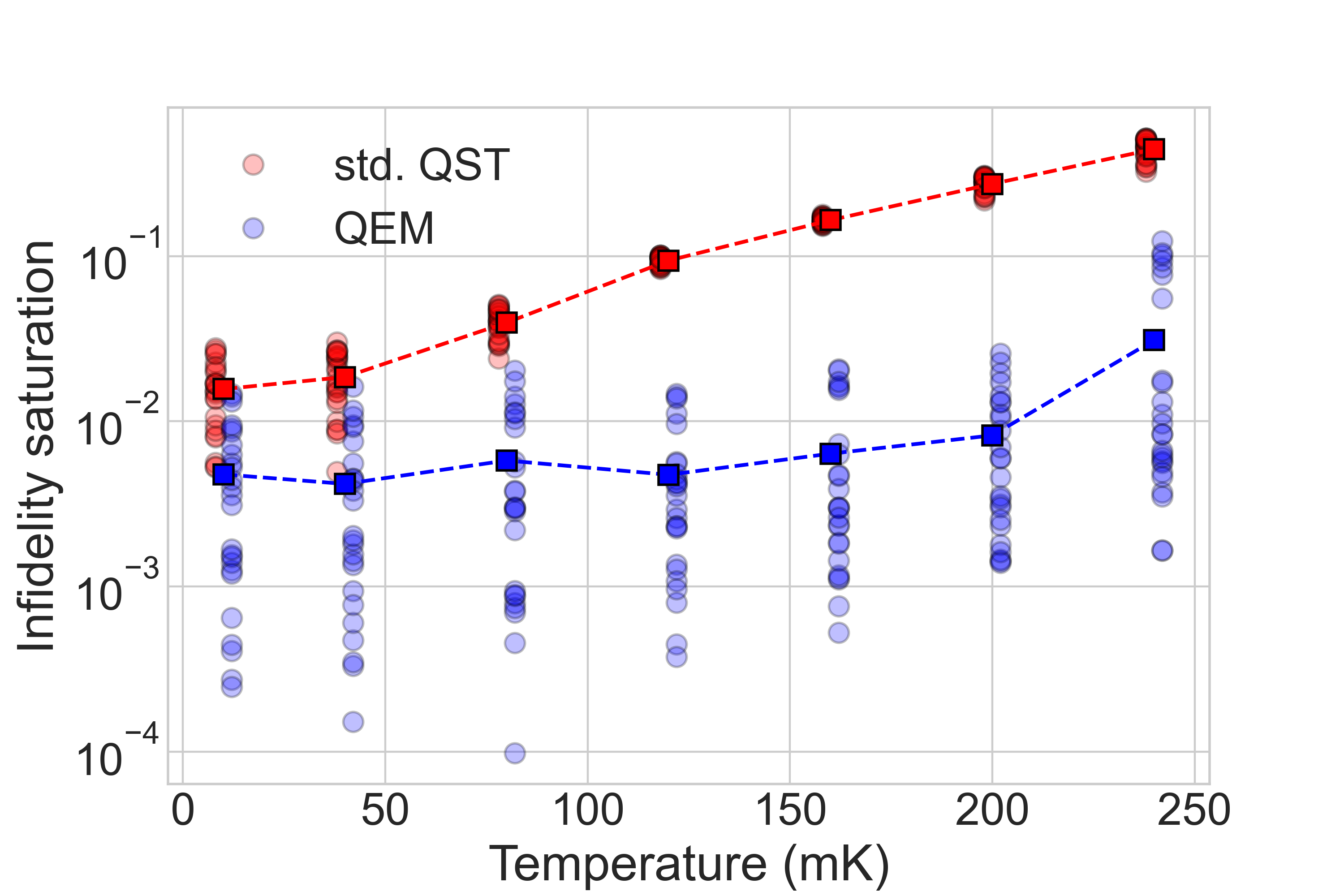 }
    \caption{\justifying 
Error mitigation applied to increased qubit temperatures. Mitigation seems to be successfully, but the assumption of perfect state preparation is broken. Averaged over 25 Haar-random pure states with otherwise optimal experimental parameters. The infidelity saturation is the infidelity averaged over the random states after 240k experimental shots.}
    \label{fig:infidelity-increased-temperature}
\end{figure}

The results in Fig.~\ref{fig:infidelity-increased-temperature} suggest that the mitigated QST is resilient to higher temperatures. The reason why we seemingly obtain a successful mitigation is that our benchmarking method does not allow us to distinguish state preparation errors from readout errors. Since our benchmarking method does not contain any additional gates between preparation and readout, the errors acquired in state preparation are interpreted as readout errors and the protocol manages to mitigate the errors accordingly. We emphasize, however, that this is only strictly true if the combined effect of state preparation and readout errors can be viewed as a single effective error channel that is independent of the prepared state. While this is the case for finite qubit temperature, which can be modeled as depolarizing noise, it is not the case in general. In particular, for the cases of qubit detuning and increased $T_1$ and $T_2$ time, the putative effective noise channel becomes dependent on the prepared state, leading to a deterioration of the mitigation efficiency. We discuss this in more detail in Appendix \ref{App:State_dependent_errors}.


\subsection{Experimental drift}
After performing QDT, our protocol assumes that the POVM stays fixed for the remainder of the measurement sequence. However, since experimental parameters drift over time, the physically realized POVM may change with respect to the POVM used for state reconstruction. Due to this drift, one would ideally recalibrate the POVM before each reconstruction. This entails a relatively large overhead, and is not feasible. 

The drift present in our experiment was small and we decided not to warrant any additional corrective measures. This is not necessarily the case in general. In these cases, we propose adding another step to the protocol in Sec.~\ref{sec:measurement_protocol}, 4: periodically perform drift measurements and recalibration. 

Drift measurements amount to measuring a set of well-known states, e.g. Pauli states, and reconstruct their density matrix. If the reconstruction infidelity goes beyond an acceptable threshold $I(\rho^{\text{estm}}, \rho)\geq \epsilon$, one can perform a recalibration of the measurement device, i.e. repeat the QDT step. This would ensure that the accuracy of the protocol does not degrade.

\section{Conclusion and outlook}
\label{sec:conclusion_and_outlook}

We have presented a comprehensive scheme for readout error mitigation in the framework of quantum state tomography. It introduces quantum detector tomography as an additional calibration step, with a small overhead cost in the number of experimental samples. After calibration, our method is able to mitigate any errors acquired at readout. Furthermore, it does not require inversion of any error channels and guarantees that the final state estimate is physical. Comparing to most previously discussed REM methods, our protocol is able to correct beyond-classical errors. To confirm that such errors indeed make up a significant part of the errors in our experiment, in Appendix \ref{appendix:Insufficiency_of_classical_error_mitigation}, we present a selection of reconstructed POVM elements from the experiment. The significant off-diagonal contributions in this analysis confirm that non-classical errors are always present and often on the same order of magnitude as the diagonal classical redistribution errors.

To probe the limits of readout error mitigation, we applied our protocol to a superconducting qubit system. We experimentally subjected the qubit to several noise sources and investigated the protocol's ability to mitigate them. We observed an improvement of the readout quality by decreasing infidelity by a factor of 5 to 30 depending on the type readout noise. The protocol was particularly effective for lowered signal amplification and decreased resonator readout power compared to standard QST. We combined multiple noise sources in an experiment where conventional state reconstruction saturated early on, whereas our method was able to precisely reconstruct the quantum state. This opens up new possibilities for systems with noisy readout where accurate knowledge of the quantum state is required. For noise sources which do not exclusively affect the readout stage the protocol did not perform optimally, and we observed a constant ratio between the infidelity of mitigated and unmitigated state reconstruction.

Potential limitations of the scheme were investigated and we presented prescriptions on how to overcome them. Overall, we found that, by using readout error mitigation, one obtains accurate state estimates even under significantly degraded experimental conditions, making the readout more robust.

While being limited by exponential scaling in both memory and required number of measurements, we expect the protocol to be feasible for up to 5-6 qubits if one replaces BME with MLE in the state reconstruction. This represents an interesting domain for error mitigation in quantum simulation as low order correlators of a larger system offer relevant information such as correlation propagation \cite{Richerme2014} and phase estimation \cite{Ebadi2021}. Recent developments in scalable approaches involving overlapping tomography \cite{Cotler2020, Tuziemski2023} could provide a framework for a scalable version of this protocol to large qubit numbers, which we intend to explore. 
In future work, it could be interesting to perform a similar experiment on multiqubit systems. Implementing the protocol on a different qubit architecture with different sources of noise would be a topic of further interest. Another option of interest is to investigate adaptive noise conscious strategies within this framework \cite{IvanovaRohling2023,Huszr2012}.  

\textit{Code \& Data availability.}
The code developed for this project is available on GitHub: \url{https://github.com/AdrianAasen/EMQST}. A short tutorial notebook is provided with examples of how to run the software. It can be interfaced with an experiment, or run as a simulation.

The experimental data are available upon request from A. Di Giovanni.

\textit{Author contributions.}
Development of the protocol and software was done by A. Aasen with supervision from M. Gärttner. The experimental realization and data generation was done by A. Di Giovanni with supervision from H. Rotzinger and A. Ustinov. A. Aasen and A. Di Giovanni prepared the draft for the manuscript. 
All authors contributed to the finalization of the manuscript. 

\acknowledgments

The authors are grateful for the quantum circuit provided
by D. Pappas, M. Sandberg, and M. Vissers. We thank W. Oliver and G. Calusine for providing the parametric amplifier.
This work was partially financed by the Baden-Württemberg Stiftung gGmbH. The authors acknowledge support by the state of Baden-Württemberg through bwHPC
and the German Research Foundation (DFG) through Grant No INST 40/575-1 FUGG (JUSTUS 2 cluster).

\iftrue
\appendix

\section{Explicit protocol realization}
\label{appendix:explicit_protocol}

 We present the explicit implementation of the protocol used in Sec.~\ref{sec:Results}. A pseudo-code outline of the whole measurement and readout error mitigation is presented in Algorithm \ref{Alg:REM_pseudocode}. For QDT, we use the maximum likelihood estimator described in Ref.~\cite{Fiurek2001}. We follow the prescription described in Sec.~\ref{subsec:quantum_detector_tomography}, and use all of the Pauli states as calibration states. The number of times each Pauli state is measured equals the maximal number of shots used for a single spin measurement in the state reconstruction, such that the dominant source of shot noise is not QDT. For QST, we use a BME, in particular we use the implementation described in Ref.~\cite{Struchalin2016}. The bank particles are generated from the Hilbert-Schmidt measure \cite{Zyczkowski2001}. This QST is also equipped with adaptive measurement strategies, which is a possible future extension to the current protocol. 

For both QDT and QST we use the Pauli-6 POVM, which is a static measurement strategy. It is known that the asymptotic scaling of such strategies depends on the proximity of the true state to one of the projective measurements \cite{Bagan2006,Struchalin2018}. To counteract this, an average over Haar-random states \cite{Mezzadri2006} is performed to get a robust performance estimate. In this way, we get the expected performance given that no prior information about the true state is available. We emphasize that for all random state reconstructions, the same QDT calibration is used.

\begin{algorithm}[H]
	\caption{Outline of readout error mitigation protocol} 
\textbf{Input}:\\ $\{\rho_B\}$: List of calibration states for QDT\\
$\{\rho_S\}$: (Haar-random) list of unknown states to measure.\\
$\{\{ M_i\}\}$: A set of POVMs forming an IC POVM, e.g basis measurement in $\sigma_x, \sigma_y, \sigma_z$\\
\textbf{Output}:\\ $\{\rho_{\text{REM}}\}$: List of readout error mitigated reconstructed states \\
$\{\{\tilde M_i\}\}$: A set of reconstructed noisy POVMs\\
$\{I\} $: List of reconstruction infidelites between readout error mitigated QST and standard QST\\\\
\textbf{Experiment:}
\begin{algorithmic}[1]
    \For {each measurement basis $\rho_B$}
        \For {each POVM $\{ M_i\}$}
            \State $\text{QDT}_\text{data} \leftarrow  \text{Measure}(\rho_B, \{ M_i\})$
        \EndFor
    \EndFor
    \For { each Haar-random state $\rho_S$}
        \For {each POVM $\{ M_i\}$}
            \State $\text{QST}_\text{data} \leftarrow  \text{Measure}(\rho_s, \{ M_i\}) $
        \EndFor
    \EndFor
\end{algorithmic}
\, \\
\textbf{Post-processing}:
\begin{algorithmic}[1]
    \For {each POVM in IC POVM}
        \State $\{\tilde M_i\} \leftarrow \text{QDT} (\text{QDT}_\text{data},\rho_B)$ \cite{Fiurek2001}
    \EndFor
    \For {each unknown state $\rho_S$}
        \State Construct modified likelihood function $\mathcal{L}_{\tilde M}(\rho)$, eq. \eqref{eq:Likelihood_function} 
        \State Construct standard likelihood function $\mathcal{L}_{M}(\rho)$
        \If {BME}
            \State Integrate $\mathcal{L}_{\tilde M}$ and find mean state $\rho_\text{BME}$ \cite{Struchalin2016} 
            \State Integrate $\mathcal{L}_{ M}$ and find mean state $\tilde \rho_\text{BME}$
        \EndIf
        \If {MLE}
            \State Find $\rho_\text{MLE}$ that maximizes $\mathcal{L}_{\tilde M}$ \cite{Lvovsky2004}
            \State Find $\tilde \rho_\text{MLE}$ that maximizes $\mathcal{L}_{ M}$
        \EndIf
        \State $\{\rho_\text{REM}\} \leftarrow \rho_\text{MLE/BME}$
        \State $\{I\} \leftarrow \text{Inf}(\rho_\text{MLE/BME}, \tilde \rho_\text{MLE/BME})$
    \EndFor

\end{algorithmic} 
\label{Alg:REM_pseudocode}
\end{algorithm}

\section{Convergence of modified estimator}
\label{appendix:Estimator_convergence}
Here we show that the protocol estimator in eq.~\eqref{eq:Adapted_likelihood_function} converges to the noiseless state with the modifications presented. By convergence of the estimator, we mean that the expected state converges to the true state for all quantum states, in the limit of infinite statistics. To make this explicit, we show that an estimator with the standard likelihood function in eq.\,\eqref{eq:Likelihood_function}  converges to any given true state $\rho_T$. The generalization to noisy measurements follows immediately.  

For simplicity, we will assume that the measurement device used implements an IC and minimal POVM. By extension, any quantum state is uniquely defined by the probabilities $p_i=\Tr(\rho M_i)=\langle M_i \rangle$. The natural estimator for the probabilities is the frequency of each outcome,
\begin{equation}
    \hat{p}_i=\frac{n_i}{N}.
\end{equation}
All we need to show, is that using this estimator for the probabilities maximizes the likelihood function in eq.~\eqref{eq:Likelihood_function}.
To be explicit, we need to show that  $\hat{p}_i=\tfrac{n_i}{N}$ is converges to the probability $p_i=\Tr(\rho M_i)$. This can be done by maximizing the log-likelihood function $\log(\mathcal{L}_M(\rho))$ with Lagrange multipliers. This yields that $p_i=\hat{p}_i$ maximizes the likelihood function. Therefore, if the experiment implements the POVM $\{M_i\}$, which means it samples from the probabilities $p_i=\Tr(\rho_TM _i)$, it is clear that we have a reconstruction $\hat \rho$ that converges to $\rho_T$. 

Consider the case where the measurement device is imperfect and implements the noisy POVM $\{\Tilde{M}_i\}$, and we have the reconstructed POVM from the QDT $\{\Tilde{M}_{i}^{\text{estm}}\}$. Analogous to the noiseless case, $\hat{\Tilde p}_i=\tfrac{\Tilde n_i}{N}$ converges to $\tilde p_i=\Tr(\rho_T \Tilde M_i)$, as long as $\langle \tilde{M}_{i}^{\text{estm}} \rangle$ converges to $ \Tilde{M}_i$. This is guaranteed by using the MLE reconstruction outlined in Ref.~\cite{Fiurek2001}.

\section{State preparation errors as effective state-dependent readout error channels}
\label{App:State_dependent_errors}

We present a straightforward model demonstrating how unitary gate errors can create an effective state-dependent error channel during state preparation, and how this influences the readout error mitigation protocol. By state-dependent, we mean that the effective error channel introduced during state preparation depends on specific properties of the prepared state. Our examination will be restricted to the case of a single qubit.

Consider a channel that describes the preparation of basis state $\rho_S$ from the ground state $\ket{0}$.
\begin{equation}
    \mathcal{E}_S(\ket{0}\bra{0}) = \rho_S.
\end{equation}
In the operator-sum representation a general channel can be represented as \cite{Nielsen2012} 
\begin{equation}
    \mathcal{E}(\rho) = \sum_i K_i \rho K_i^\dagger,
\end{equation}
where $\sum K_i^\dagger K_i = \mathbb{1}$.
When the state preparation is noiseless, the operator-sum representation is particularly simple
\begin{equation}
    \mathcal{E}_{S, \text{ideal}}(\rho) = U_S \rho U_S^\dagger,
\end{equation}
where $U_S$ unitary that implements the transformation $\ket{0}\bra{0} \rightarrow \rho_S$. We can decompose any single-qubit unitary \cite{Nielsen2012}
\begin{equation}
 U_S = R_Z(\theta_2) R_Y(\theta_1) R_Z(\theta_0), \label{eq:Unitary_decomposition} 
\end{equation} 
where $R_{\alpha}(\theta) = e^{-i \theta \sigma_{\alpha}/2}$ are the single qubit rotation operators about the Cartesian axis $\alpha$ of the Bloch sphere. We will now consider the channel $\mathcal{E}_S$ where with probability $p_i$ implements a unitary $\tilde U_{S,i}$  which deviates from the ideal $U_S$. We want to show that $\mathcal{E}_S = \mathcal{E}_{S,\text{error}} \circ 
 \mathcal{E}_{S, \text{ideal}}$. 
\begin{prop}
The state preparation channel can be separated into $\mathcal{E}_S = \mathcal{E}_{S,\text{error}} \circ 
 \mathcal{E}_{S, \text{ideal}}$, where $\mathcal{E}_{S,\text{error}}$ generally depends on the state $S$ being prepared.
\end{prop}
\begin{proof}
Our strategy is to show that each $\tilde U_{S,i} = \tilde U_{S,i,\text{error}} U_{S}$, where $\tilde U_{S,i,\text{error}}$ is some unitary, that may or may not depend on $S$. From such a decomposition of each $i$ it immediately follows that the channels also decompose. We start by expressing $\tilde U_{S,i}$ in the decomposition in eq.~\eqref{eq:Unitary_decomposition}  (where we drop the index $i$ for convenience) 
 \begin{equation}
     \tilde U_{S} = R_Z(\tilde \theta_2) R_Y(\tilde \theta_{1}) R_Z(\tilde \theta_0).
 \end{equation}
 This channel will always be applied to the ground state $\ket{0}$, so we can ignore  $R_Z(\tilde \theta_0)$ as it does not affect the ground state.  The remaining angles can be expressed in terms of the ideal ones as
 \begin{equation}
     \tilde \theta_k = \theta_k + \Delta_k.
 \end{equation}
 Since the rotation gate commutes with itself, we can split them
 \begin{equation}
     R_\alpha(\tilde \theta_k) = R_\alpha(\theta_k) R_\alpha (\Delta_k) = R_\alpha (\Delta_k) R_\alpha(\theta_k).
 \end{equation}
We can then write the the noisy unitary as 
\begin{equation}
\begin{aligned}
    \tilde U_{S} =&  R_Z(\tilde \theta_2) R_Y(\tilde \theta_{1})\\
    =& R_Z(\Delta_2) R_Z(\theta_2)  R_Y(\Delta_{1}) R_Y(\theta_{1})\\
    =& \underbrace{R_Z(\Delta_2) R_Z(\theta_2)  R_Y(\Delta_{1}) R_Z^\dagger(\theta_2) }_{U_{S,\text{error}}} \underbrace{R_Z(\theta_2)  R_Y(\theta_{1})}_{U_{S}}\\
    =& U_{S,\text{error}} U_{S}.
    \label{eq:noisy_unitary_decomposition}
\end{aligned}
\end{equation}
The error unitary generally depends on the rotation that was performed through $\theta_2$. 

Consider now the channel $\mathcal{E}_S$ where we with some probability $p_i$ implement the unitary $\tilde U_{S,i}$,
\begin{equation}
    \mathcal{E}_S(\rho) = \sum_i p_i \tilde U_{S,i} \rho \tilde  U_{S,i}^\dagger
    \label{eq:noise_operator_sum_channel}
\end{equation}
Inserting the noisy unitary from eq.~\eqref{eq:noisy_unitary_decomposition} into the channel we have
\begin{equation}
\begin{aligned}
    \mathcal{E}_S(\rho) =& \sum_i p_i U_{S,\text{error}} \overbrace{U_{S} \rho  U_{S}^\dagger}^{\mathcal{E}_{S, \text{ideal}}(\rho)} U_{S,\text{error}}^\dagger\\
    =& \sum_i p_i U_{S,\text{error}} \mathcal{E}_{S, \text{ideal}}(\rho) U_{S,\text{error}}^\dagger\\
    =&  \mathcal{E}_{S,\text{error}} \circ 
 \mathcal{E}_{S, \text{ideal}}(\rho)
    \end{aligned}
\end{equation}
and we have showed that the noise channels separate.
\end{proof}

This separation of the state preparation channel means that $\mathcal{E}_{S,\text{error}}$ will be interpreted as part of the readout errors in the passive picture.  Importantly, we see the explicit dependence of the error channel on the prepared state and the distribution $p_i$. 
With a state-dependent error channel present, one cannot reconstruct a consistent set of POVM elements based on the $I \cross S$ constraints, since each set of constrains per calibration state ($S$) comes from a different channel. 
This means, that the POVM reconstructed in the detector tomography stage, which is based on a finite set of calibration states, will in general not faithfully represent the error incurred for the calibration states themselves or any other prepared states.

With this simple and experimentally motivated model we can provide an explanation of why the protocol still works for some cases, while in others it does not. Note that the operator-sum representation is not unique, so it could still be possible to find representations where the explicit dependence vanishes, see e.g.\ increased temperature. 



\subsection{Qubit detuning}
Qubit detuning has non-trivial contributions from both $\Delta_1$ and $\Delta_2$. We therefore expect there to be state-dependent errors present in state preparation, and subsequently our readout error mitigation protocol to under-perform. 

\subsection{Increased temperature}
The primary source of error at elevated temperatures stems from the excitation of the qubit ground state. A secondary contribution arises from unitary gate errors, with finite temperatures interpreted as an error symmetrically distributed over fluctuating $\Delta_1$ and $\Delta_2$. The symmetric convex combination of these unitary rotations results in a depolarizing error channel that does not depend on the prepared state, thereby enabling us to transfer these state-preparation errors to the readout stage as well. 

\subsection{Shorter $T_1$ and $T_2$ times}
Increased manipulation times cannot easily be viewed as effective unitary gate errors. However, it is plausible that also non-unitary errors can lead to a state-dependent effective error channel. The protocol's under-performance during extended manipulation times might suggest that state-dependent errors are indeed induced.



\section{Digitizing experimental data}
\label{app:experimental calibrations}

\begin{figure}[t]
    \includegraphics[width=.9\linewidth]{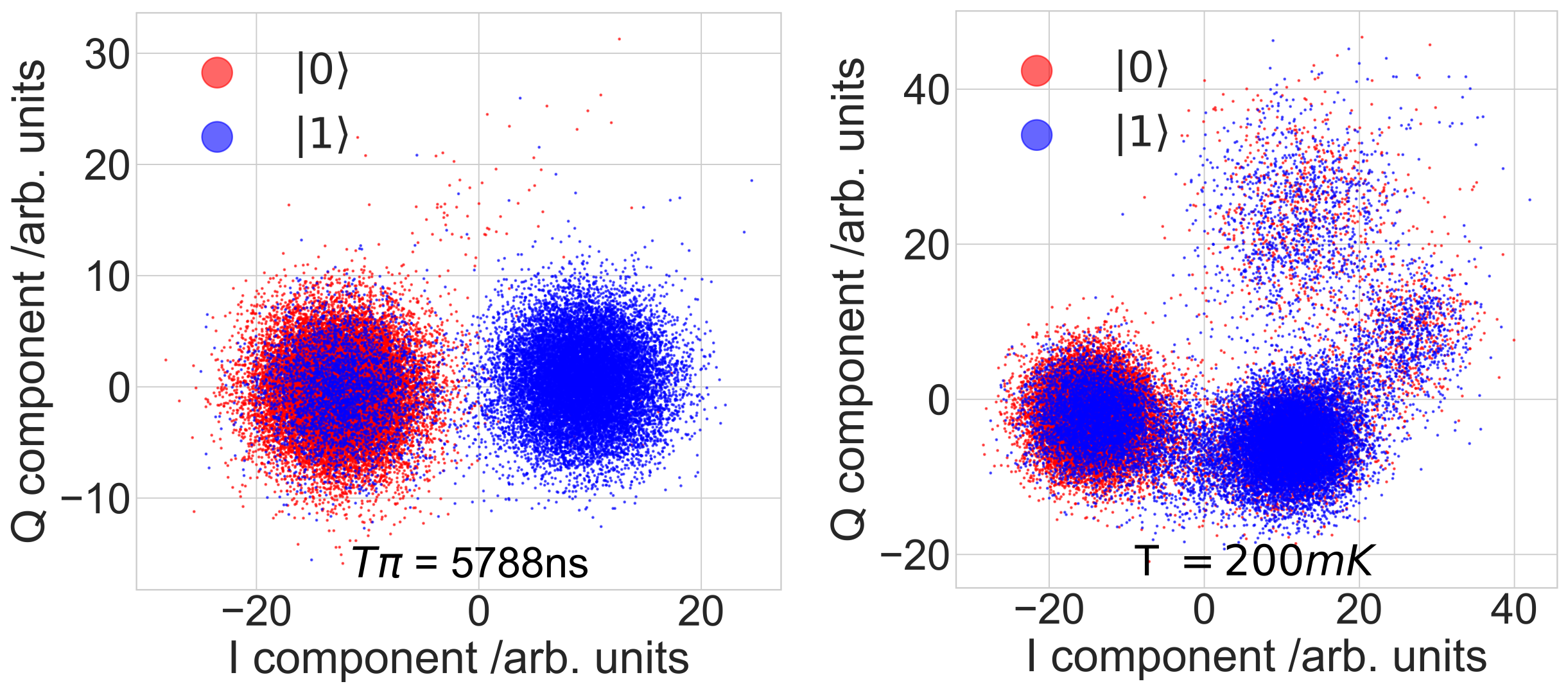}
    \caption{\justifying Single-shot qubit measurements in the IQ plane, color-coded to the two qubit states. Long manipulation pulses allow the qubit to decay more, populating the ground state more often (left).  Higher temperatures shift the readout to higher states, populating out-of-qubit states, too (right). }
    
    \label{fig:IQ-scatter}
\end{figure}

After measuring the length of a $\pi$-pulse, we use a supervised classification algorithm on the data to digitize the measured analog voltages (I and Q, see Fig.~\ref{fig:IQ-scatter}) into a single bit: 0 or 1. Often, two 2D Gaussian distributions can be fitted on the measured data points, but using machine-learning-based classification turned out to work more efficiently than Gaussian fitting, particularly for strong noises. We used the Python package sklearn, which runs in an intermediate level between the low-level experimental drivers and the high-level Bayesian and MLE algorithms.

\section{Insufficiency of classical error mitigation}
\label{appendix:Insufficiency_of_classical_error_mitigation}
To demonstrate that classical errors are not sufficient to capture all error present in the measurement device, we investigate coherent errors in the POVM elements. Classical errors assume that the reconstructed POVM elements are purely diagonal in their measurement basis, i.e.\ only contains a statistical redistribution of the diagonal elements. Coherent errors manifest in the off-diagonal elements and are beyond classical. For a superconducting qubit the Pauli-6 POVM is performed by measuring individually the three bases of the Pauli-operators. We can then investigate how off-diagonal each of these bases are with respect to their ideal eigenbasis. 

In Fig.~\ref{fig:coherence-errors}  we present the reconstructed POVMs for some of the noisy runs in their ideal measurement basis (i.e. we present the POVM elements $R_{X\rightarrow Z} M_X R_{X\rightarrow Z}^\dagger$, $R_{Y\rightarrow Z} M_Y R_{Y\rightarrow Z}^\dagger$ and $M_Z$). For a single qubit it is sufficient to consider only the element corresponding to the outcome 0.

\begin{figure*}[htp]
\centering
\includegraphics[width=.4\textwidth]{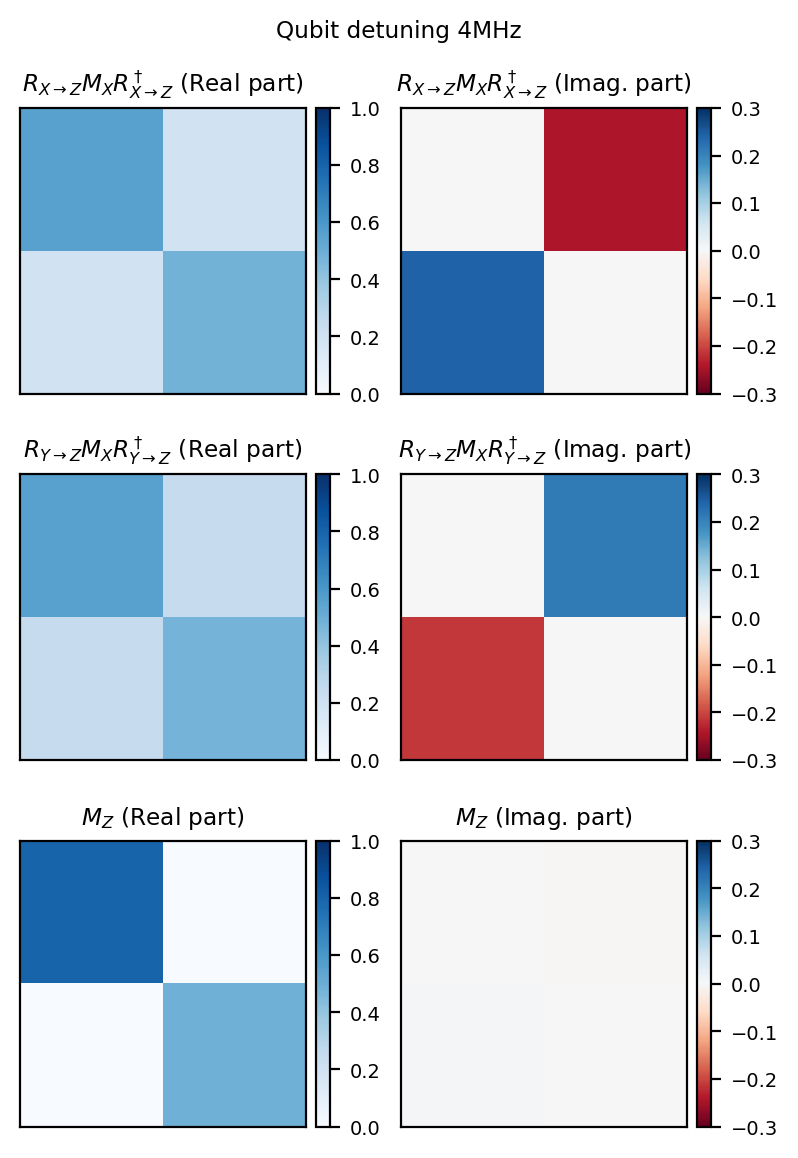}
\includegraphics[width=.4\textwidth]{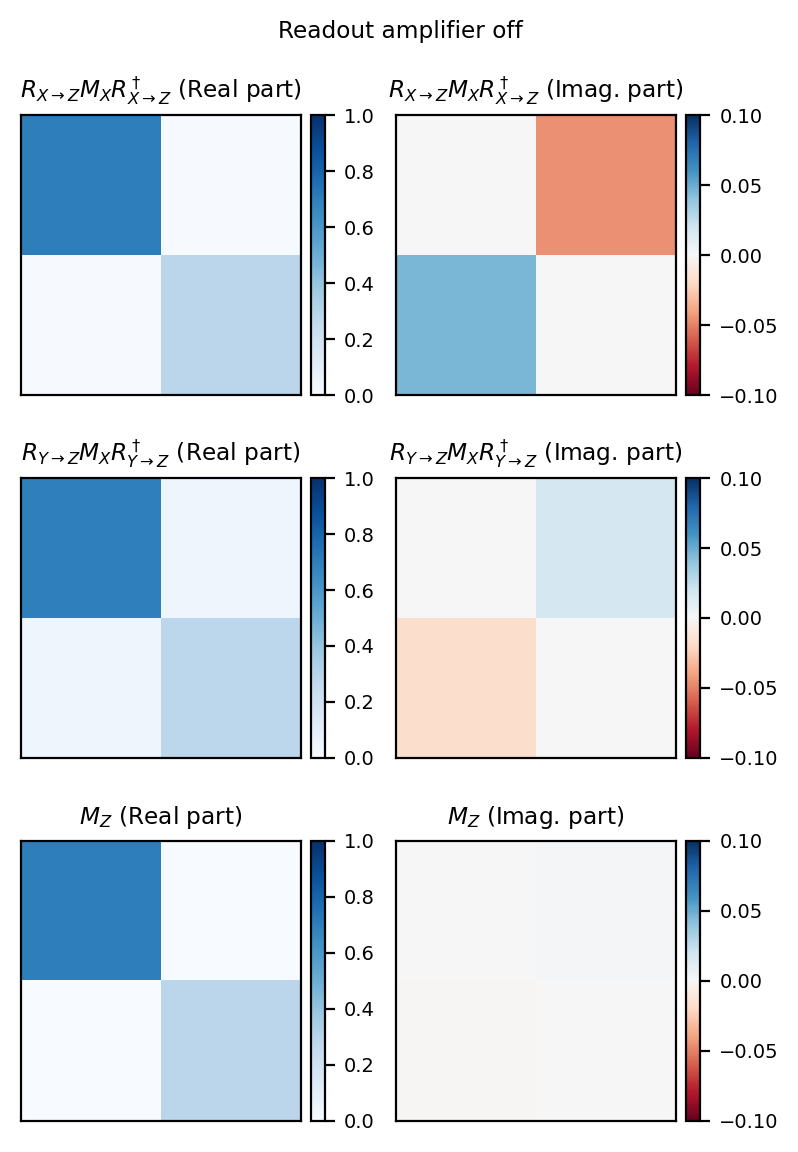}
\includegraphics[width=.4\textwidth]{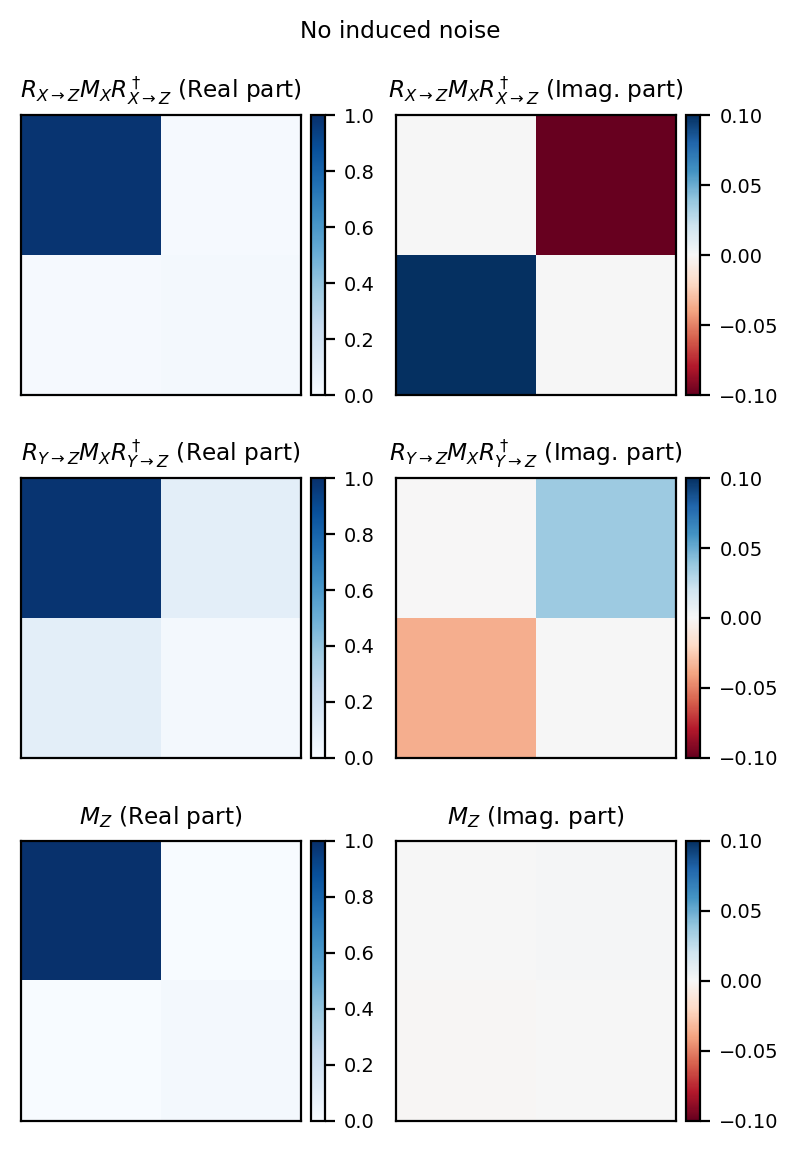}
\includegraphics[width=.4\textwidth]{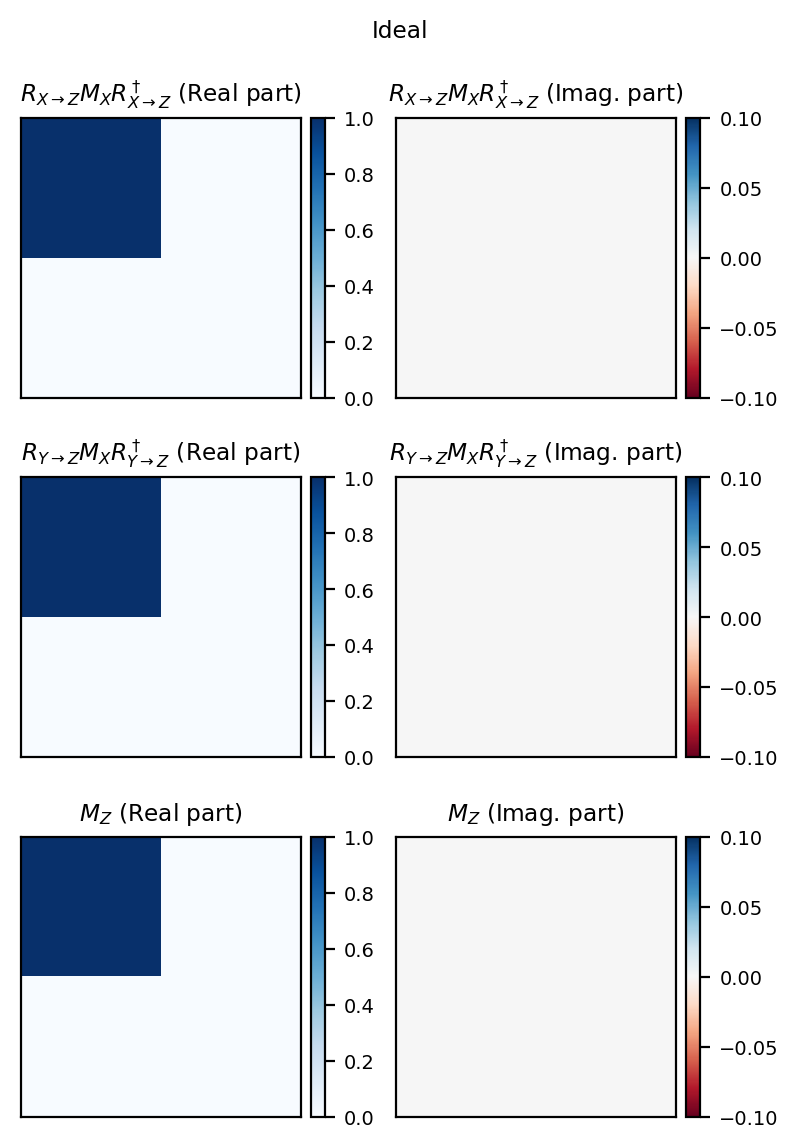}
\caption{The POVM elements in their ideal eigenbasis. Non-trivial off-diagonal elements signal coherent errors that are not captured by classical error mitigation. The plots shows, starting from the top left: Qubit detuning at 4 MHz, turning off the parametric amplifier, optimal operation of the experiment, and theoretical ideal POVM. The POVMs were estimated using $N=240$k single-shot measurements for each computational basis, which 
 means statistical fluctuations are on the order of $\propto 3\times 10^{-2}$, well below the coherence errors observed. }
\label{fig:coherence-errors}
\end{figure*}

We found that for almost all cases, the z-basis measurement (the computational basis measurement) does not have coherence errors beyond statistical fluctuations, while for many instances the x- and y-basis measurements suffer from large coherence errors in the induced noise scenarios. 


\bibliography{refs}{}
\bibliographystyle{ieeetr}
\end{document}